\documentclass[11pt,letterpaper]{article}

\usepackage[margin=1in]{geometry}
\usepackage{hyperref}
\usepackage{verbatim}
\usepackage{amsmath,amsthm,amssymb,amsfonts,amscd,amstext}
\usepackage{graphicx}
\usepackage[table]{xcolor}
\usepackage{multirow} 
\usepackage{times} 
\usepackage{mathtools}
\usepackage{url}
\usepackage{esint}
\usepackage{bbm}

\usepackage{sidecap}
\sidecaptionvpos{figure}{c} 

\newcommand{\blind}[2]{#1}

\usepackage{rotating}

\global\long\def\R{\mathbb{R}}

\usepackage[font=scriptsize]{caption}

\newcommand{\recomb}[2]{#2}

\usepackage{setspace}
\usepackage[ruled, vlined]{algorithm2e}

\usepackage{tikz}
\usetikzlibrary{decorations.markings}
\newcommand{\dhorline}[3][0]{%
    \tikz[baseline]{\path[decoration={markings,
      mark=between positions 0 and 1 step 4*#3
      with {\node[fill, circle, minimum width=#3, inner sep=0pt, anchor=south west] {};}},postaction={decorate}]  (0,#1) -- ++(#2,0);}}
\newcommand{\dottedhfill}{\noindent\makebox{\dhorline{0.93\linewidth}{0.8pt}}}

\global\long\def\BC{Barber and Cand\`es}
\global\long\def\KR{Katsevich and Ramdas}

\global\long\def\subKR{_{\text{KR}}}
\global\long\def\til#1{\tilde{#1}}%
\global\long\def\alp{\alpha}
\global\long\def\gam{\gamma}
\global\long\def\del{\delta}
\global\long\def\eps{\varepsilon}%
\global\long\def\sig{\sigma}%
\global\long\def\betaa{\boldsymbol{\beta}}%
\global\long\def\epss{\boldsymbol{\eps}}%
\global\long\def\yy{\boldsymbol{y}}%
\global\long\def\N{\mathbb{N}}%
\global\long\def\R{\mathbb{R}}%
\global\long\def\supsec{Supplementary Section}%
\global\long\def\supfig{Supplementary Figure}%

\theoremstyle{plain}
\newtheorem{Theorem}{Theorem}
\newtheorem{Assumption}{Assumption}

\newtheorem{Lemma}{Lemma}

\newtheorem{Corollary}[Lemma]{Corollary}

\theoremstyle{remark}
\newtheorem{Remark}{Remark}

\newcommand{\eqq}{\coloneqq}

\newenvironment{alglist}
  {\begin{list}
    {}
    {\setlength{\labelwidth}{1em}
     \setlength{\labelsep}{0em}
     \setlength{\itemsep}{2pt}
     \setlength{\leftmargin}{1cm}
     \setlength{\rightmargin}{1.2cm}
     \setlength{\itemindent}{0em} 
     
    }
  }
{\end{list}}

\newcounter{alglistcounter}

\begin{document}

\title{Competition-based control of the false discovery proportion}

\blind{\author{Dong Luo$^1$, Arya Ebadi$^1$, Kristen Emery$^1$, Yilun He$^1$,\\William Stafford Noble$^2$, Uri Keich$^1$\\
$^1$School of Mathematics and Statistics F07\\
University of Sydney\\
$^2$Departments of Genome Sciences and of Computer Science and Engineering\\
University of Washington\\
}}{}

\maketitle

\begin{abstract}
Recently, \BC\ laid the theoretical foundation for a general framework for false discovery rate (FDR) control based on the notion of ``knockoffs.''
A closely related FDR control methodology has long been employed in the analysis of mass spectrometry data, referred to there as ``target-decoy competition'' (TDC).
However, any approach that aims to control the FDR, which is defined as the expected value of the false discovery proportion (FDP),
suffers from a problem. Specifically, even when successfully controlling the FDR at level $\alpha$,
the FDP in the list of discoveries can significantly exceed $\alpha$.
We offer FDP-SD, a new procedure that rigorously controls the FDP in the competition (knockoff / TDC) setup
by guaranteeing that the FDP is bounded by $\alpha$ at any desired confidence level.
Compared with the just-published general framework of \KR, FDP-SD generally delivers more power and often substantially so
in simulated as well as real data.
\end{abstract}

%




\section{Introduction}

Competition-based false discovery rate (FDR) control has been widely practiced by the computational mass spectrometry community since it was
first proposed by Elias and Gygi \cite{elias:target,cerqueira:mude,jeong:false,elias:target2,granholm:determining,the:talk}.
Consider for example the spectrum identification (spectrum-ID) problem where our goal is to assign for each of the, typically,
tens of thousands of spectra the peptide that has most likely generated it (\supsec \ref{sec:background} provides further context).

Spectrum-ID is typically initiated by scanning each input spectrum against a peptide database for its
best matching peptide.
Pioneered by SEQUEST~\cite{eng:approach}, the search engine uses an elaborate score function to quantify the quality of the match between
each of the database peptides and the observed spectrum, recording the optimal peptide-spectrum match (PSM) for the
given spectrum along with its score $Z_i$~\cite{nesvizhskii:survey}.
In practice, many expected fragment ions will fail to be observed for any given spectrum, and the spectrum is also likely to contain a
variety of additional, unexplained peaks~\cite{noble:computational}.
Hence, sometimes the reported PSM is correct --- the peptide assigned to the spectrum was
present in the mass spectrometer when the spectrum was generated --- and sometimes the PSM is incorrect.
Ideally, we would report only the correct PSMs, but obviously we do not know which PSMs are correct and which are incorrect;
all we have is the score of the PSM, indicating its quality.
Therefore, we report a thresholded list of top-scoring PSMs while trying to control the list's FDR using
target-decoy competition (TDC), as explained next.

First, the same search engine is used to assign each input spectrum a \emph{decoy}
PSM score, $\til Z_{i}$, by searching for the spectrum's best match in a decoy
database of peptides obtained from the original ({\em target}) database by randomly
shuffling or reversing each peptide in the database.
Each decoy score $\til Z_{i}$ then directly competes with its corresponding
target score $Z_{i}$ to determine the reported list of discoveries, i.e.,
we only report target PSMs that win their competition: $Z_{i}>\til Z_{i}$.
Additionally, the number of decoy wins ($\til Z_{i}>Z_{i}$) in the top $k$ scoring PSMs is used to
estimate the number of false discoveries in the target wins among the same top $k$ PSMs.
Thus, the ratio between the number of decoy wins and the number of target wins
yields an estimate of the FDR among the target wins in the top $k$ PSMs. 
To control the FDR at level $\alp$, the TDC procedure (Supplementary Section \ref{supsec:algs})
chooses the largest $k$ for which the estimated FDR is still $\le\alpha$,
and it reports all target wins among those top $k$ PSMs.
It was recently shown that, assuming that \emph{incorrect PSMs are independently equally likely to come
from a target or a decoy match}, and provided we add 1 to
the number of decoy wins before dividing by the number of target wins, this procedure
rigorously controls the FDR~\cite{he:theoretical}.

More recently, \BC\ used the same principle in their knockoff+ procedure to control the FDR in feature selection in a classical linear
regression model~\cite{barber:controlling}\recomb{.}{:\begin{equation}
\yy=X\betaa+\epss,\label{eq:lin_model}
\end{equation}
where $\yy\in\R^{n}$ is the response vector, $X$ is the $n\times p$
known, real-valued design matrix, $\betaa\in\R^{p}$ is the unknown
vector of coefficients, and $\epss\sim N(0,\sig^{2}I)$ is Gaussian
noise. 
Briefly, knockoff+ relies on introducing an $n\times p$ {\em knockoff}
design matrix $\til X$, where each column consists of a knockoff
copy of the corresponding original variable. These knockoff variables
are constructed so that in terms of the underlying regression problem
the true null features (the ones that are not included in the model)
are in some sense indistinguishable from their knockoff copies. The
procedure then assigns to each null hypothesis $H_{i}:\beta_{i}=0$
two test statistics $Z_{i},\til Z_{i}$ which correspond to the point
$\lambda$ on the Lasso path \cite{tibshirani:regression}
at which feature $X_{i}$, respectively, its knockoff competition $\tilde{X}_{i}$, first enters the
model. The intuition here is that generally $Z_{i}>\til Z_{i}$ for
true model features, whereas for null features, $Z_{i}$ and $\til Z_{i}$
are identically distributed.

}
While \BC' knockoff construction is significantly more elaborate than
that of the analogous decoys in the spectrum-ID problem, knockoffs and decoys serve the same purpose in competition-based FDR control.
Moreover, following their work and the introduction of a more flexible formulation of the variable selection problem
in the model-X framework of Cand\'es et al.~\cite{candes:panning}, competition-based FDR control
has gained a lot of interest in the statistical and machine learning communities, where it
has been applied to various applications in biomedical research~\cite{xiao:mapping,gao:model,read:predicting} and has been extended
to work in conjunction with deep neural networks~\cite{lu:deeppink} and time series data~\cite{fan:ipad},
as well as to work in a likelihood setting without requiring the use of latent variables~\cite{sudarshan:deep}.

FDR control is a popular approach to the analysis of multiple testing.
However, it should not be
confused with controlling the false discovery proportion (FDP). The latter is the proportion of true nulls
among all the rejected hypotheses (discoveries), and the FDR is its expectation\recomb{.}{ (taken with respect to the true null hypotheses).}
In particular, while controlling the FDR at level $\alpha$, the FDP in any given sample can exceed $\alpha$.
Thus, in practice, controlling the FDP is arguably more relevant than the FDR in most cases.
Figure \ref{fig:intro} provides examples from both the mass spectrometry (left) and feature selection from linear regression (right)
domains showing that while the FDR is controlled (left: $\alp=0.05 > 0.047=\overline{\text{FDP}}$ , right: $\alp=0.2 > 0.198=\overline{\text{FDP}}$)
the FDP can significantly exceed $\alp$.


\begin{figure}
\centering %
\begin{tabular}{ll}
Peptide detection  & Feature selection \tabularnewline
\includegraphics[width=3in]{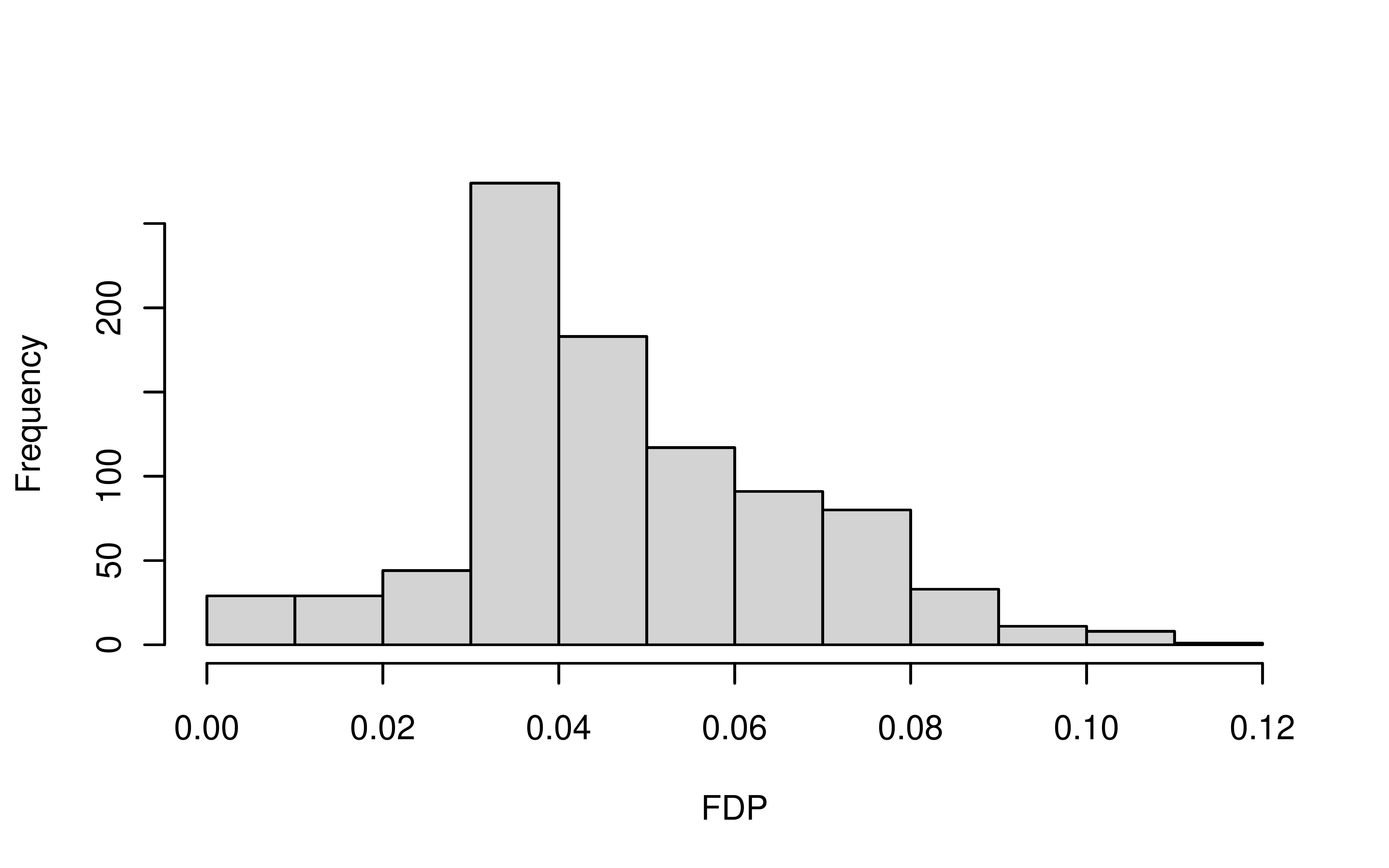}  & \includegraphics[width=3in]{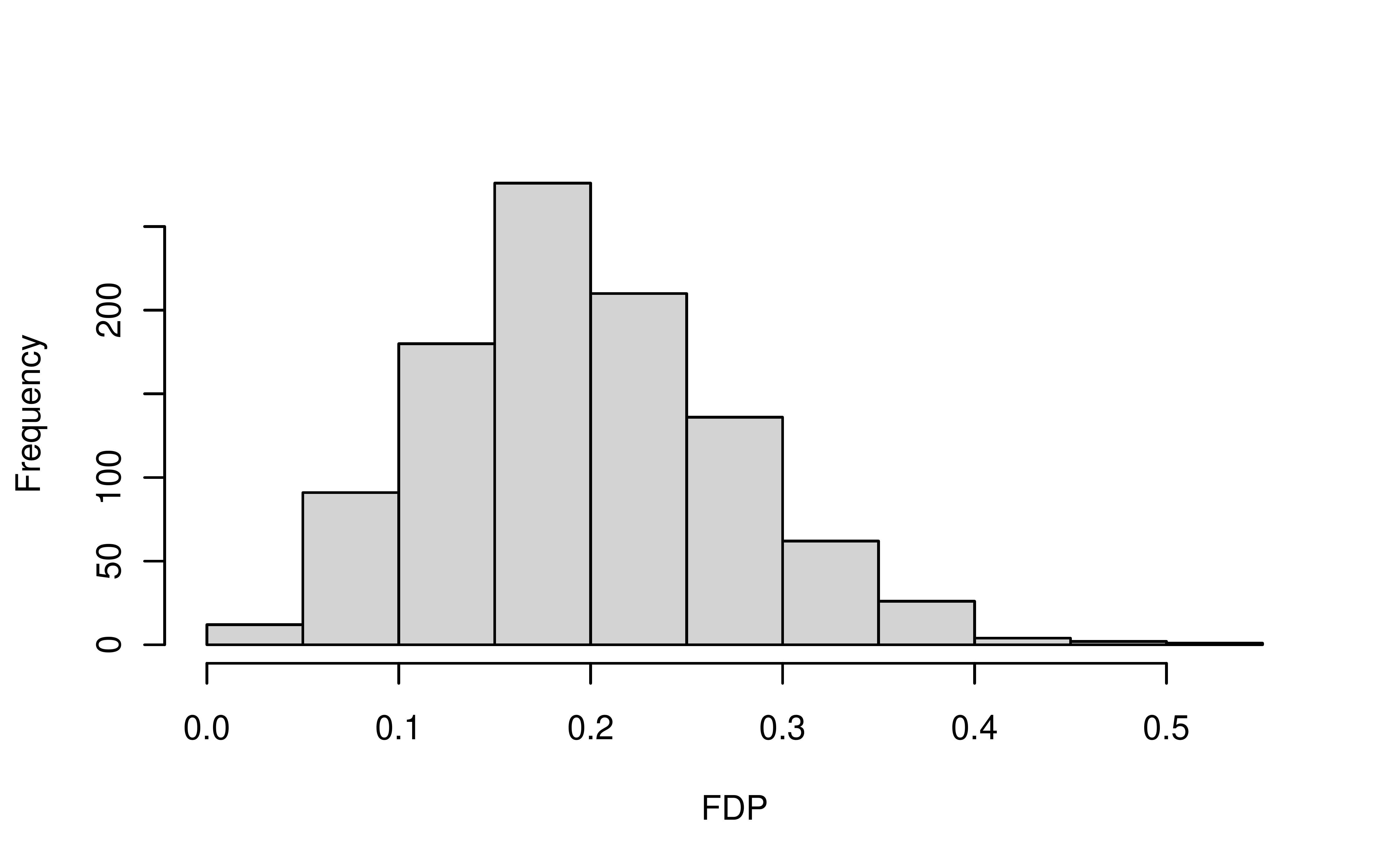} \tabularnewline
\end{tabular}
\caption{\textbf{Controlling the FDR does not imply controlling the FDP.} (A) We applied TDC to control the FDR (at $\alp=0.05$)
among the peptides detected in runs of the ISB18 data (see Section \ref{section_applications} for details).  Because ISB18 comes from a controlled experiment
we inferred the FDP for each of those 900 TDC runs and generated the presented histogram. Note that the average of the FDP,
$\overline{\text{FDP}}=0.047$, is indeed $<\alp=0.05$.
(B) The histogram was generated by 1000 applications of the model-X knockoff procedure to select the relevant features in that many linear
regression models while controlling the FDR at $\alp=0.2$ and noting the actual FDP among the selected features (see Section \ref{section_applications} for details). Note that
$\overline{\text{FDP}}=0.198$, is indeed $<\alp=0.2$.\label{fig:intro}}
\end{figure}

When introducing the notion of FDR, Benjamini and Hochberg noted that, strictly speaking,
the FDP cannot be controlled at any non-trivial level~\cite{benjamini:controlling}. Indeed, imagine that all the hypotheses
are true nulls: rejecting any hypothesis would then imply the FDP is 1.
Nevertheless, some notion of controlling the FDP, or false discovery exceedance (FDX) control, has been extensively studied in the canonical setup
of multiple hypothesis testing where p-values are available (e.g., \cite{genovese:stochastic,genovese:exceedance,lehmann:generalizations}).
An FDP-controlling procedure in this context reports a list of discoveries with the guarantee that $Q$, the FDP among the
reported discoveries, is bounded by the threshold $\alp$ with high confidence: $P(Q>\alpha)\le\gamma$, where
$1-\gamma$ is our desired confidence level.

Here we offer a practical procedure that rigorously controls the FDP in the target-decoy / knockoff competition context
(we mostly stick to the target-decoy terminology, but our analysis is also applicable to the knockoffs).
Specifically, given the desired confidence level $1-\gamma$ and the FDP threshold $\alpha$ (in addition to the target and decoy scores)
our novel ``FDP stepdown'' (FDP-SD) procedure yields a list of target discoveries so that $P(Q>\alpha)\le\gamma$.
That is, when using FDP-SD the FDP can still be larger than the desired $\alpha$;
however, now that probability is bounded by $\gamma$.
Like TDC, FDP-SD's reported discoveries consist of all target wins among the top $k$ scoring PSMs.
FDP-SD also shares with TDC the use of the observed number of decoy wins in the top scores to obtain a bound on
the number of unobserved false target wins. Specifically, as the ``stepdown'' in its name suggests, FDP-SD finds
the rejection threshold by sequentially comparing the number of decoy wins
to pre-computed bounds and stopping with the first index for which the corresponding bound is exceeded.

\KR\ very recently developed a general framework for obtaining simultaneous upper confidence bounds
on the FDP that applies to our competition based setup~\cite{katsevich:simultaneous}. In particular, their approach can be used to provide a competing procedure
to FDP-SD which we refer to as FDP-KRB. In Section \ref{section_applications} we provide extensive evidence that FDP-SD, which
was independently developed, generally offers more power than FDP-KRB, and often substantially so.

FDP-SD is available for download at \url{https://github.com/uni-Arya/stepdownfdp}.

\section{The model}\label{sec:model}

In this section we lay out the assumptions that our analysis relies on (see Supplementary Table \ref{suptable:definitions} for a summary of our notation).
Let $H_i$ ($i=1,\dots,m$) denote our $m$ null hypotheses, e.g., in the spectrum-ID problem $H_i$ is ``the $i$th PSM is incorrect,''
and in the linear regression problem $H_i$ is ``the coefficient of the $i$th feature is 0.''
Associated with each $H_i$ are two competing scores: a target/observed score $Z_i$ (the higher the score
the less likely $H_i$ is) and a decoy/knockoff score $\til Z_{i}$. For example, in the spectrum-ID problem $Z_i$ ($\til Z_{i}$)
is the score of the optimal target (decoy) peptide match to the $i$th spectrum\recomb{.}{, whereas in linear regression
$Z_{i}$ ($\til Z_{i}$) correspond to the point $\lambda$ on the Lasso path at which the feature (its knockoff) entered the model.}

Adopting the notation of \cite{emery:multipleRECOMB} we associate with each hypothesis a score $W_i$ and a target/decoy-win label $L_i$.
By default $W_i=Z_i\vee \til Z_i$ (i.e., the max of the two scores as in spectrum-ID) but as \BC\ pointed out, other functions such as $W_i=|Z_i - \til Z_i|$ can be
considered as well. As for $L_i$:
\[
L_i=\begin{cases} 
      1 & Z_i > \til Z_i \text{ ($H_i$ corresponds to a target/original feature win)} \\
      0 & Z_i = \til Z_i \text{ (tie, $H_i$ is ignored)} \\
      -1 & Z_i < \til Z_i \text{ ($H_i$ corresponds to a decoy/knockoff win)} \\
   \end{cases} .
\]
Because $H_i$ is ignored if $L_i=0$ without loss of generality we assume that $L_i\ne0$ for all $i$.

Let $N=\{i\,:\,H_i \text{ is a true null}\}$ and note that while typically in the context of hypotheses testing $N$ is a constant,
albeit unknown set, it is beneficial here to allow $N$ to be a random set as well.
Our fundamental assumption is the following:
\begin{Assumption}\label{fund_assump}
Conditional on all the scores $\{W_i\}_i$ and all the false null labels $\{L_i\,:\,i\notin N\}$,
the true nulls are independently equally likely to be a target or a decoy win, i.e.,
the random variables (RVs) $\{L_i\,:\,i\in N\}$ are conditionally independent uniform $\pm1$ RVs.
\end{Assumption}
Clearly, if the assumption holds then $\{L_i\,:\,i\in N\}$ are still independent uniform $\pm1$ RVs after ordering
the hypotheses in decreasing order so that $W_1\ge W_2\ge\dots\ge W_m$.

Some specific competition paradigms that satisfy Assumption \ref{fund_assump} include
\begin{itemize}
	\item
	the theoretical model of TDC introduced by He et al.~\cite{he:theoretical}: their assumptions of ``equal chance'' and ``independence''
	are an equivalent formulation of our assumption;
	\item
	the original FX (fixed design matrix $X$) knockoff scores construction of \BC~\cite{barber:controlling}: see Lemma 1.1 and
	its ensuing discussion, keeping in mind that our $W_i$ is their $|W_i|$ and $L_i$ is the sign of their $W_i$; and
	\item
	the MX (random design matrix) knockoff scores of Cand\`es et al.~\cite{candes:panning}: see Lemma 2 (same notation comment as for FX).
	\item
	the spectrum-ID model proposed in~\cite{keich:improved,keich:controlling} and briefly described in \supsec~\ref{sec:background}.
\end{itemize}

Our list of reported discoveries consists of all target wins among the top $k$ scores for some $k$.
Therefore, without loss of generality we assume our hypotheses are ordered in decreasing order of $W_i$, and our goal is
to analyze the following random variables/processes (for $i=0$ we set all counts to 0):
$D_i=\sum_{j=1}^i1_{\{L_j=-1\}}$ (the number of decoy wins in the top $i$ scores),
$T_i=\sum_{j=1}^i1_{\{L_j=1\}}$ (the corresponding number of target wins; no ties implies $D_i+T_i=i$),
and $N^t_i := \sum_{j=1}^i 1_{\{L_j=1,j\in N\}}$ (the number of true null target wins in the top $i$ scores).
With this notation, the FDP among all target wins in the top $i$ scores is $Q_i = N_i^t / (T_i\vee 1)$.

While the spectrum-ID model is captured by Assumption \ref{fund_assump}, as noted in \supsec~\ref{sec:background},
there are a couple of features that are distinct to this problem. First, the set $N$ of true null hypotheses is random,
and second, a false null (correct PSM) has to correspond to a target win.
This is not the case in general. For example, in the feature selection problem, a feature is a false null when its coefficient in
the regression model is not zero.  It is possible for such a feature to have a lower score than its corresponding
knockoff and hence to be counted as a decoy win.

Finally, here we assumed that all target-decoy ties ($L_i=0$) are thrown out, but if instead we randomly break ties then Assumption \ref{fund_assump} still holds.
In our practical analysis we randomly broke ties.
Similarly, how the $W_i$ are sorted in case of ties should not matter as long as that ordering is independent of the corresponding labels.

\section{Controlling the FDP}

\subsection{\KR' approach to FDP control}
\label{sec:FDP-KRB}

A stochastic process $\{\xi_i\}_{i\in I}$ is a $1-\gamma$ \emph{upper prediction band} for the random process $Z_i$ with $i\in I\subset\N$
if $P(\exists i\in I\,:\,Z_i > \xi_i)\le \gamma$. \KR\ recently developed a general framework for constructing such bands that, as they showed,
can be specialized to construct an upper prediction band $\{\xi_i\}$ on $V_d\coloneqq N^t_{i_d}$ (the number of true null target wins before the $d$th decoy win).
As pointed out by \KR, this band can be used to control the FDP by reporting all target wins among the top $k_{\subKR}$ scores, where
\begin{equation}
k_{\subKR}=\max\left\{ k\le m\,:\,\xi_{D_{k}+1}/T_{k}\le\alpha\text{ or }k=0\right\} \label{eq:tau_FDP_control}.
\end{equation}
We refer to this procedure, summarized as Algorithm \ref{algorithm:FDP-KRB} in \supsec~\ref{supsec:algs}, as FDP-KRB.

\subsection{FDP-SD: a novel approach to FDP control via stepdown}
\label{sec:FDP-SD}

Originating in the canonical context where p-values are available, stepdown procedures work by sequentially comparing the $i$th smallest p-value, $p_{(i)}$, against
a precomputed bound $\del_i$. Specifically, the procedure looks for $k_{\text{SD}}=\max\left\{i\,:\,p_{(j)}\le\del_j\text{ for }j=1,2,\dots,i\right\}$
and rejects the corresponding $k_{\text{SD}}$ hypotheses~\cite{lehmann:generalizations}.

FDP-SD is inspired by a stepdown procedure developed by Guo and Romano to control the FDP when p-values are available~\cite{guo:generalized}.
Because we have no p-values in our competition context, we instead use the number of decoy wins: FDP-SD sequentially goes through the hypotheses
sorted in order of decreasing scores, comparing the observed number of decoy wins $D_i$ with precomputed bounds $\del_{\alpha,\gamma}(i)$
that depend on the desired FDP threshold $\alp$ and the confidence level $1-\gamma$.

The bounds $\del_{\alpha,\gamma}(i)$ are set to allow us to control the FDP when rejecting all target wins in the
top $i$ scores for a fixed $i$. Specifically, imagine we report all target wins in the top $i$ scores if $D_i\le\del_{\alpha,\gamma}(i)$, and
otherwise we report none. Then $\del_{\alpha,\gamma}(i)$ should be sufficiently large so that
regardless of the number of true nulls among the top $i$ scores, 
the probability that the FDP among our reported discoveries exceeds $\alp$ is $\le\gam$.
To ensure optimality of the bound we also require
that the same cannot be guaranteed for any bound greater than $\del_{\alpha,\gamma}(i)$.
It is not difficult to show that this requires us to define $\del_{\alpha,\gamma}(i)$ as:
\begin{equation}
\label{eq:delta_definition}
\delta_{\alp,\gam}(i) \coloneqq \max \left\{d\in\{-1,0,1,\dots,i\}\,:\,F_{B(k(d)+d,1/2)}(d) \le \gam \right\},
\end{equation}
where $k(d) = k(i,d)\coloneqq \lfloor(i-d)\alpha\rfloor + 1$ and $F_{B(n,p)}(\cdot)$ denotes the cumulative distribution function (CDF) of a
binomial $B(n,p)$ RV so $F_{B(k(d)+d,1/2)}(d)=P[B(k(d)+d,1/2)\le d]$.\recomb{}{

The intuition here is that if there are $k(d)$ or more false (true null) discoveries, then the FDP exceeds $\alp$ so we make sure
that the probability there were $k(d)$ or more true null target wins is bounded by $\gam$. The reason we can do this is that the total number
of true nulls in the top $i$ scores is bounded by $k(d)+d$ (it is exactly this for the spectrum-ID problem) and each true null is independently equally
likely to be a target or a decoy win. Hence, the unobserved number of false target wins is stochastically bounded by a $B(k(d)+d,1/2)$ RV.
}

Typically, $\delta_{\alp,\gam}(i)=-1$ for small values of $i$. Indeed, with $i=1$
it is impossible to get any confidence $1-\gam>1/2$ that the corresponding hypothesis is not a true null target win.
Therefore, we should only compare $D_i$ with  $\delta_{\alp,\gam}(i)$ when the latter is $\ge0$.
Using Lemma \ref{Lem:pmd_is_mono} in \supsec~\ref{supsec:proof_thm_FDP-SD},
which shows that $F_{B(k(d)+d,1/2)}(d)$ is increasing in $d\le i$, it is easy to see that with
\begin{equation}
	\label{eq:i0}
	i_0 = i_0(\alp,\gam)\coloneqq \max\{1, \lceil\left(\lceil\log_2\left(1/\gam\right)\rceil-1\right)/\alp\rceil\} ,
\end{equation}
$i\ge i_0$ if and only if $\del_{\alpha,\gam}(i)\ge 0$.\recomb{}{ Note also that for a fixed $\alp$ and $\gam$, $\del_{\alpha,\gam}(i)$ is non-decreasing in $i$.}

After computing $i_0$ FDP-SD finds
\begin{equation}
	\label{eq:k_FDP-SD}
	k_{\text{FDP-SD}} =  \max\Big\{i \, : \, \prod_{j=i_0}^i 1_{D_j\le\del_{\alpha,\gamma}(j)}  = 1 \text{ or } i=0 \Big\},
\end{equation}
where $1_A$ is the indicator of the event $A$, and
it reports the $T_{k_{\text{FDP-SD}}}$ \emph{target} discoveries (wins) among the top $k_{\text{FDP-SD}}$ scores.
The following theorem guarantees that the FDP-SD procedure, which is summarized in \supsec~\ref{supsec:algs}, controls the FDP.

\begin{Theorem}
\label{thm:FDP-SD}
With $k_{\text{FDP-SD}}$ defined as in \eqref{eq:k_FDP-SD} let $Q_{\text{FDP-SD}}$ be the FDP among the $T_{k_{\text{FDP-SD}}}$
target wins in the top $k_{\text{FDP-SD}}$ scores. Then $P(Q_{\text{FDP-SD}} > \alp)\le\gam$.
\end{Theorem}
The proof, inspired by that of Theorem 3.2 of \cite{guo:generalized}, is given in \supsec~\ref{supsec:proof_thm_FDP-SD}.

The bounds $\delta_{\alp,\gam}$ that FDP-SD relies on are computed in \eqref{eq:delta_definition} using binomial CDFs.
Because the binomial distribution is discrete it is typically impossible to find a $d$ for which $F_{B(k(d)+d,1/2)}(d) \le \gam$ holds with equality.
As a result, FDP-SD typically attains a higher confidence level than required: $P(Q_{\text{FDP-SD}} > \alp )< \gam$.
We address this issue by introducing a more powerful, randomized version of FDP-SD in \supsec~\ref{supsec:algs}.
The proof that the randomized version still rigorously controls the FDP is similar to the proof of Theorem~\ref{thm:FDP-SD}
so it is skipped here.

\section{Extension and Limitation}

\subsection{Extending FDP-SD to utilize multiple decoys}

Emery et al.~recently developed FDR-controlling procedures for the setup where we have $d>1$ decoys
for each hypothesis~\cite{emery:multipleRECOMB}. 
Using their framework, which is applicable when the decoys are independently generated, as well as when they
satisfy a weaker exchangeability condition \cite[\supsec~6.13]{emery:multipleRECOMB},
we can extend FDP-SD to take advantage of multiple decoys in a fairly straightforward manner.

Indeed, assume that associated with each of the $m$ hypotheses are $d$ decoys. Let $d_1 \coloneqq d + 1$ and
let $c = i_c/d_1$ and $\lambda = i_\lambda/d_1$ with $i_c, i_\lambda \in \{1,2, \dots, d\}$ be the target and decoy win thresholds
(here we regard these thresholds as predetermined tuning parameters and reserve the question of how to set them for future research).

Let $r_i \in \{1, \ldots, d_1\}$ be the rank of the target score in the combined list of the target and all decoy scores associated with hypothesis $H_i$ (with higher ranks corresponding to larger scores). As usual, we break any ties among the scores at random. Define the label $L_i$ associated with $H_i$ by
\begin{equation}\label{multi-label}
    L_i := \left\{\begin{array}{ll}
        1 & \text{if } r_i \geq d_1 - i_c + 1 \\
        0 & \text{if } r_i \in (d_1 - i_\lambda,\, d_1 - i_c + 1) \\
        -1 & \text{if } r_i \leq d_1 - i_\lambda
    \end{array}.\right.
\end{equation}
In words, if the rank of the target score is among the top $i_c$ ranks (top $(100 \cdot c)$\% ranks) we label $H_i$ as a target win, whereas
if the target rank is among the bottom $d_1 - i_\lambda$ ranks (bottom $[100 \cdot (1-\lambda)]$\% ranks) we label $H_i$ as a decoy win.
Otherwise, we ignore $H_i$ for the rest of the procedure, labeling it with $L_i = 0$.

Define the winning score $W_i$ to be the $s_i^{\text{th}}$ highest ranked score for hypothesis $H_i$, where
\begin{equation}\label{multi-selected-rank}
    s_i := \left\{\begin{array}{ll}
        r_i & \text{if } L_i = 1 \\
        u_i & \text{if } L_i = 0 \\
        \varphi(r_i) & \text{if } L_i = -1
    \end{array}.\right.
\end{equation}
Here, $u_i$ is a (uniformly chosen) random element of $\{d_1-i_c+1,\, \ldots,\, d_1\}$,
and $\varphi : \{1,\, \ldots,\, d_1-i_\lambda\} \rightarrow \{d_1-i_c+1,\, \ldots,\, d_1\}$ is a map of
\textit{losing ranks} (those for which $L_i = -1$) into \textit{winning ranks} (those for which $L_i = 1$).
In words, \eqref{multi-selected-rank} says that if we have a target-winning hypothesis (that is, $L_i = 1$), 
the winning score is the target score; otherwise, the winning score is one of the decoy scores among the winning ranks.
The mapping $\varphi$, which we do not define here, is constructed so that assuming for example that the decoys are independently generated,
the rank $r_i$ of a true null target score is distributed uniformly in $\{1, \ldots, d_1\}$.
The formal definition of this mapping is given in \cite{emery:multipleRECOMB}
but two common choices are the max mapping, $\varphi\equiv d_1$, and the mirror mapping, $\varphi(j)=d_1-j+1$.
Note that this extends the single decoy case, where a truly null hypothesis is required to be a target or decoy win with equal probability (in this
case there is only one possible mapping function).

Once we labeled the hypotheses and computed the winning scores, we apply a slightly generalized version of FDP-SD that is adapted to make use
of the multiple decoys (see Algorithm \ref{algorithm:FDP-SDm} and its randomized version, Algorithm \ref{algorithm:r-FDP-SDm}, in the supplement).

The proof that, for a predetermined choice of $c$ and $\lambda$, both these procedures control the FDP in the resulting list of discoveries
is almost identical to that of Theorem \ref{thm:FDP-SD}. The key difference is that the probability of observing a decoy-winning true null,
given that it was counted, is no longer $1/2$, as in the single decoy case, but instead
\begin{equation*}
    R = R(c, \lambda) := \frac{1 - \lambda}{c + 1 - \lambda}.
\end{equation*}

\subsection{Non-Admissibility of FDP-SD}\label{non-admissible}

In Section \ref{section_applications} below we demonstrate that FDP-SD is generally more powerful than FDP-KRB
and hence, to the best of our knowledge, it is generally the optimal tool for controlling the FDP in the
knockoff/TDC setup.
Still, we next provide evidence that even the randomized version of FDP-SD could potentially be further improved.
Specifically, we show that there exists a valid FDP-controlling procedure $\mathcal{M}$ that uniformly improves on the latter:
it never returns fewer discoveries than the randomized FDP-SD and there exists a specific setup in which it returns more discoveries
with positive probability. In that sense even the randomized FDP-SD is non-admissible.

We define the procedure $\mathcal{M}$ as follows: it agrees with FDP-SD except when $m = 21$, $\alpha = 0.1$, and $\gamma = 0.25$,
and the labels $L_i$ corresponding to the decreasing winning scores $W_i$ satisfy $L_{20}=-1$ and $L_i=1$ for $i=1,\dots,19,21$.
In this scenario $\mathcal{M}$ reports all 20 target wins as discoveries, i.e., with $k_\mathcal{M}$ denoting $\mathcal{M}$'s
cutoff, $k_\mathcal{M}=21$ in this case.

If we let $k_{\text{r-FDP-SD}}$ denote the cutoff of the randomized FDP-SD, then clearly $k_{\text{r-FDP-SD}}\le k_\mathcal{M}$
always holds. Moreover, it is easy to see that given the same set of labels (which, for example, is attained with positive
probability if all hypotheses are true nulls) $k_{\text{r-FDP-SD}}=19<21=k_\mathcal{M}$ with positive probability
(indeed, $\bar\delta_{20} = 0$ with probability 2/3 as per Algorithm \ref{algorithm:r-FDP-SD} in the Supplementary).

Finally, recall that $\mathcal{M}$ only differs from the randomized FDP-SD when
$L_{20}=-1$ and $L_i=1$ for $i=1,\dots,19,21$ and $\bar\delta_{20} = 0$.
If $N_{19}$, the number of true nulls among the top scoring 19 hypotheses, is $\ge2$ then in this scenario
we already have $Q_{k_{\text{r-FDP-SD}}}=Q_{19}>\alp$.
Conversely, $N_{19}\le1$ in which case the number
of null target wins that $\mathcal{M}$ reports in this scenario is bounded by 2, and as it reports 20 target discoveries,
$Q_{\mathcal{M}}=Q_{21}\le\alp$.
Thus, the FDP in $\mathcal{M}$'s list of discoveries exceeds $\alp$ only when the same applies to the randomized FDP-SD and
since the latter controls the FDP so does $\mathcal{M}$.

\section{Applications to Real and Simulated Data}
\label{section_applications}

To evaluate the procedures presented here we looked at their performance on simulated and real data where competition-based FDR
control is already an established practice: 
simulated spectrum-ID and peptide detection in mass spectrometry (\supsec~\ref{sec:background}),
feature selection in linear regression, and a previously published application of the knockoff methodology
to genome-wide association studies (GWAS).
In each case our model, and specifically Assumption \ref{fund_assump}, either explicitly holds
or is believed to be a reasonable approximation.

The spectrum-ID model is presented in \supsec~\ref{sec:background}. Here we used a variant of this model described in \cite{keich:progressive}
for which, as explained in \supsec~\ref{supsec:model_simulation},
Assumption \ref{fund_assump} is only approximately valid: for native spectra there is a slightly larger chance
that a true null will be a decoy win (which creates a slightly conservative --- and hence not overly concerning  --- bias).
We generated simulated instances of the spectrum-ID problem using both calibrated and uncalibrated scores
while varying $m$, the number of spectra, among 500, 2k, and 10k and varying $\pi_0$, the proportion of foreign spectra,
among 0.2, 0.5 and 0.8. For each of these 18 data-parameter combinations we randomly drew 40K
instances of simulated target and decoy PSM scores, as described in \supsec~\ref{supsec:model_simulation}.
We then applied the considered FDR/FDP-controlling procedures to each simulated set with FDR/FDP thresholds of $\alpha$ = 1\%, 5\%, and 10\%,
and confidence levels $100(1-\gam)$ = 95\% and 99\%.

As mentioned in Section \ref{sec:model}, our model, and therefore our procedures, apply to controlling the FDR in variable selection via knockoffs.
Hence, we looked at the very first example of Tutorial 1
of ``Controlled variable Selection with Model-X Knockoffs'' (
\href{http://web.stanford.edu/group/candes/knockoffs/software/knockoffs/tutorial-1-r.html}{``Variable Selection with Knockoffs''})~\cite{candes:panning}.
Specifically, we repeated the following sequence of operations 1000 times:
we randomly drew a normally-distributed $1000\times1000$ design matrix and generated a response vector using only 60 of the 1000
variables while keeping all other parameters the same as in the online example (amplitude=4.5, $\rho$=0.25, $\Sigma$ is a Toeplitz matrix whose
$d$th diagonal is $\rho^{d-1}$).
We then computed the model-X knockoff scores (taking a negative score as a decoy win and a positive score as a target win)
and applied all the procedures at FDR/FDP levels $\alp\in\{0.1,0.2\}$ and confidence levels of $1-\gam\in\{0.90, 0.95\}$.


Our GWAS example is taken from \cite{katsevich:simultaneous}, which in turn is based on data made publicly available by
Sesia et al.~\cite{sesia:multiresolution}. The goal of this analysis was to identify genomic loci (the features) that are significant factors in the expression of
each of the eight traits that were analyzed (the dependent variables). The raw data was taken from the UK Biobank \cite{bycroft:biobank}
and transformed to a regression problem by Sesia et al., who then created knockoff statistics~\cite{sesia:multiresolution}.
We downloaded the scores using the functions \verb+download_KZ_data+ and \verb+read_KZ_data+ defined in \KR'
\href{https://raw.githubusercontent.com/ekatsevi/simultaneous-fdp/master/UKBB_utils.R}{UKBB\_utils.R}. Consistent with the latter, we
applied TDC with $\alp=0.1$ and computed upper prediction bounds using $\gam=0.05$, whereas the FDP controlling procedures used $\gamma = 0.05$ and $\alpha\in\{0.05, 0.1,0.2\}$.

When applying the procedures to our datasets we specifically looked at which of the two methods for controlling the FDP ---
FDP-KRB, and the novel FDP-SD (here we used the \emph{randomized version} described in \supsec~\ref{supsec:algs}) ---
generally delivers the most discoveries.

\begin{figure}[h]
\centering
\begin{tabular}{ll}
\hspace{-8ex}
\includegraphics[width=3.5in]{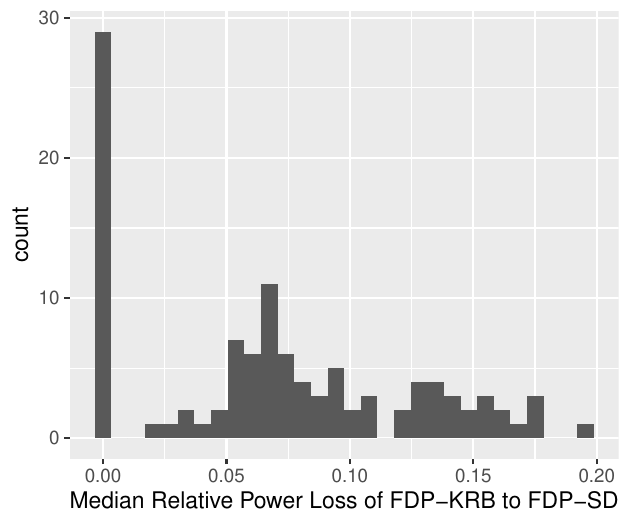} &
\includegraphics[width=2.5in]{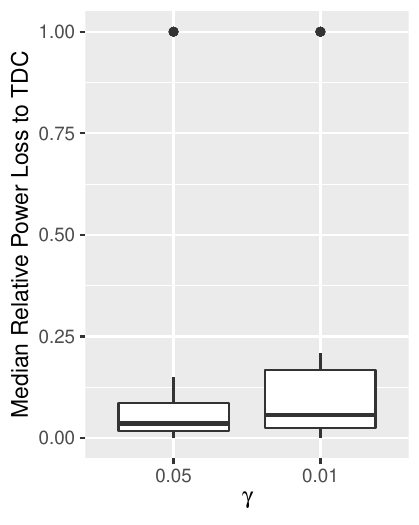}
\end{tabular}
\caption{\textbf{Simulated spectrum-ID: power loss relative to FDP-SD and relative to TDC.}\label{fig:sim_spec}\\
	Left: for each of the 108 combinations of calibrated/uncalibrated scores with $m\in\{500, 2\text{k}, 10\text{k}\}$,
	$\pi_0\in\{0.2, 0.5, 0.8\}$, $\alpha\in\{0.01, 0.05, 0.1\}$, and  $\gam\in\{0.01, 0.05\}$,
	we noted the median of the loss in power when using FDP-KRB compared with using
	FDP-SD. The median was taken over 40K samples, and the relative loss is
	defined as $1 - (T'_{\text{FDP-KRB}} + 10^{-12})/(T'_{\text{FDP-SD}} + 10^{-12})$, where $T'_{\text{FDP-*}}$
	is the number of \emph{true} discoveries reported by the method. Notably FDP-SD's median number of discoveries is
	never smaller than that of FDP-KRB across all 108 data-parameters combinations.
	The median of the 108 median power losses of FDP-KRB is 6.8\%.\\
	Right: using the same randomly generated data data we noted the median of the loss in power when using FDP-SD
	(with confidence $\gam\in\{0.01, 0.05\}$) to control the FDP compared with using TDC to control the FDR at the same level $\alp$.
	The medians of the two sets of 54 median power losses (108 combinations split according to the confidence parameter $\gam$) are:
	5.7\% ($\gam=0.01$) and 3.6\% ($\gam=0.05$).\\\\
}
\end{figure}


The left panel of Figure~\ref{fig:sim_spec} shows that in the spectrum-ID simulation FDP-KRB's median power never exceeds that of FDP-SD
with a typical power loss of about 7\% compared with the latter. We see similar results
in the GWAS example: Figure \ref{fig:BB_panel} (top-left) shows that for $\alp=0.05$ FDP-SD yields the larger number of discoveries
for all 8 traits with FDP-KRB typically yielding only 0-50\% of the number reported by FDP-SD. For $\alpha=0.1,0.2$ (middle and bottom left panels) the results are
a little more mixed: for three of the 16 trait-parameters combinations FDP-SD loses to FDP-KRB, but for the other 13 FDP-SD yields more
discoveries and typically by a wide margin.
In our linear regression data for all parameter combinations FDP-SD again reports more discoveries than FDP-KRB
(Figure~\ref{fig:lin_reg}).

\begin{figure}
\centering %
\begin{tabular}{ll}
\hspace{-5ex}
Relative power of FDP controlling procedures  & \hspace{2ex} FDP-SD's power loss relative to FDR control \tabularnewline
\hspace{-6ex}
\includegraphics[width=3.5in]{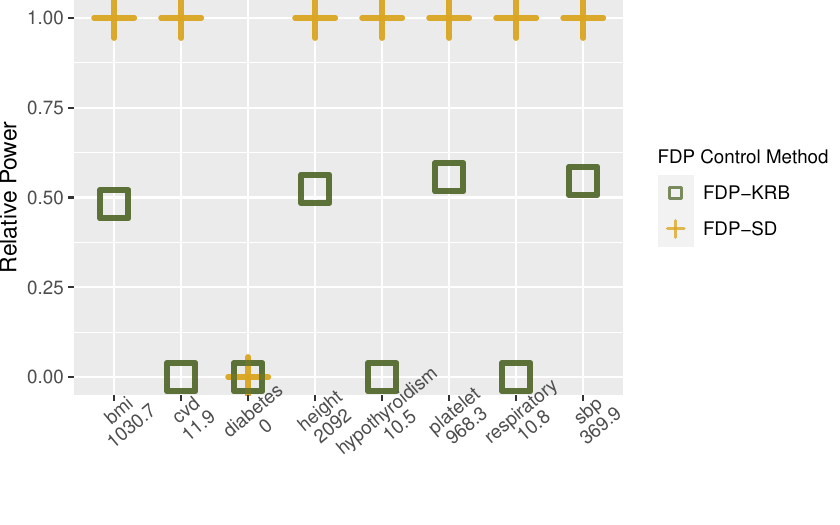}  & \includegraphics[width=3.5in]{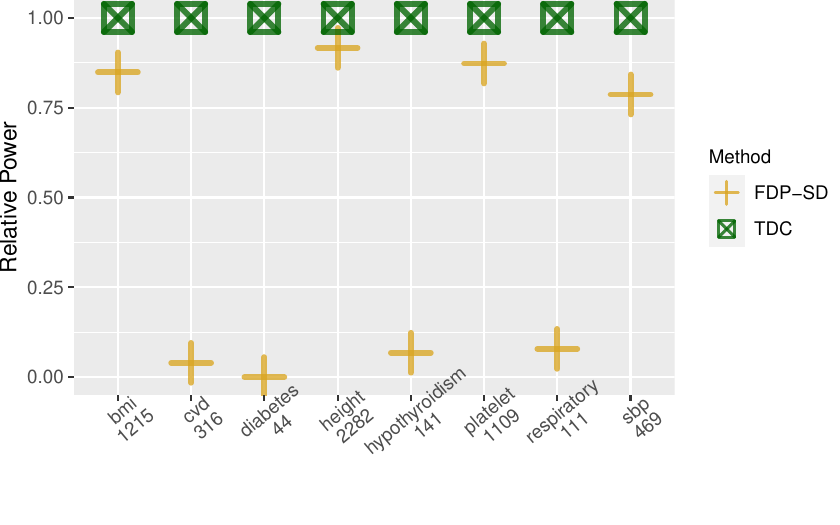} \tabularnewline
\hspace{-6ex}
\includegraphics[width=3.5in]{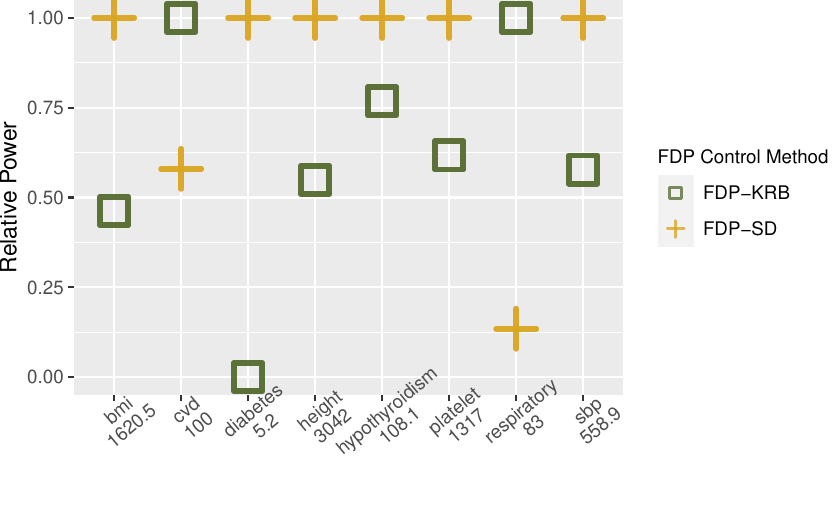} &  \includegraphics[width=3.5in]{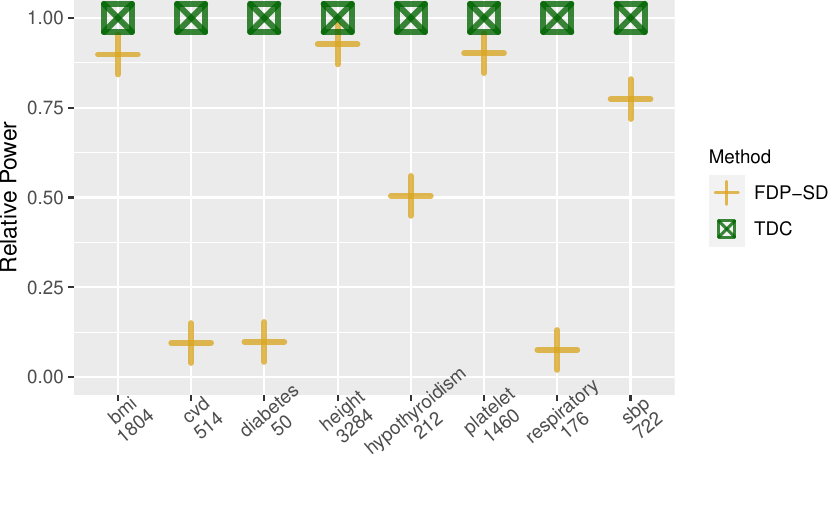} \tabularnewline
\hspace{-6ex}
\includegraphics[width=3.5in]{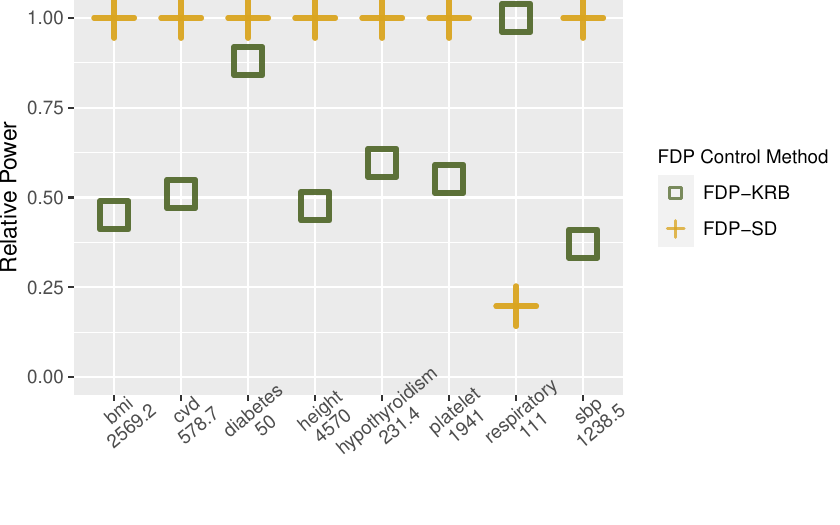}  & \includegraphics[width=3.5in]{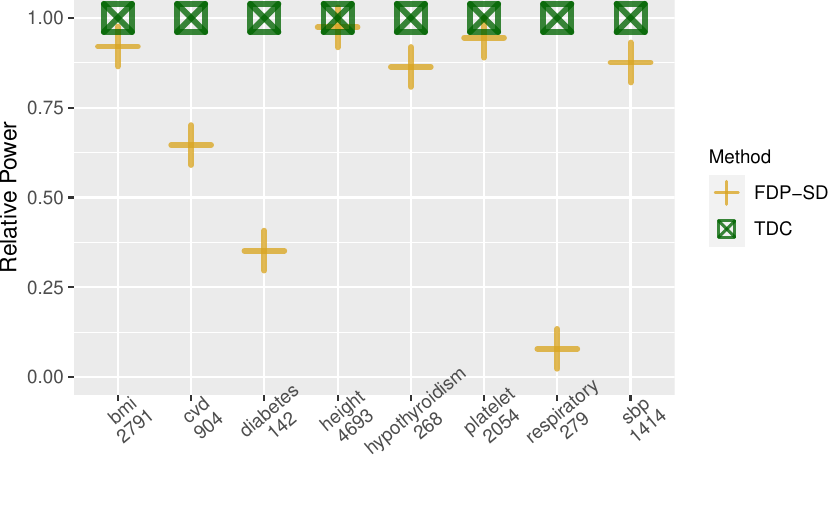} \tabularnewline
\end{tabular}
\caption{\textbf{Power of FDP controlling procedures in the GWAS example.} 
For each of the 8 investigated traits the left column panels describe the relative power of FDP-KRB and FDP-SD defined as
the ratio of the number of genomic loci the method discovers (averaged over 1K runs for FDP-SD),
while controlling the FDP, over the number reported by the optimal method for that trait-parameters combination (indicated below the trait).
The confidence level was fixed at $1-\gam=0.95$ and we varied the FDP/FDR threshold: $\alp=0.05$ (top row), $\alp=0.1$ (middle row), and  $\alp=0.2$ (bottom row).
The right column panels show the power loss when controlling the FDP (confidence $1-\gam=0.95$) using FDP-SD vs.~controlling
the FDR (same $\alp$) using TDC, where we varied the FDP/FDR threshold: $\alp=0.05$ (top row), $\alp=0.1$ (middle row), and  $\alp=0.2$ (bottom row).
\label{fig:BB_panel}}
\end{figure}

In terms of how much power is given up when controlling the FDP using FDP-SD vs.\ controlling the FDR using TDC,
the bottom-left panel of Figure~\ref{fig:sim_spec} shows in the spectrum-ID dataset that the median power loss is
3.6\% when using $\gam=0.05$, and it is 5.7\% when using $\gam=0.01$. In the GWAS example we see wide variations
in terms of power loss, depending on the trait-parameters combination: Figure~\ref{fig:BB_panel} (right column).
A similar variability is observed in the linear regression dataset (Figure~\ref{fig:lin_reg}): compare the violet mark in the TDC column
with the green ($\gam=0.05$) and violet ($\gam=0.1$) marks in FDP-SD's column, as well as the TDC's cyan mark with the corresponding red ($\gam=0.05$)
and cyan ($\gam=0.1$) marks of FDP-SD.

\begin{SCfigure}
\includegraphics[width=3.5in]{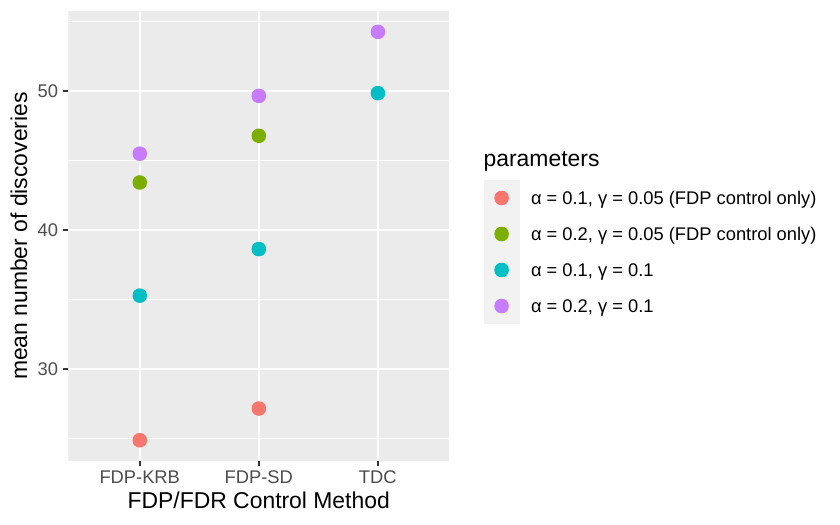}
\caption{\textbf{Linear regression: power of FDP/FDR controlling procedures.}
The figure presents the mean number of \emph{true} discoveries when controlling the FDP using 
two methods discussed in this paper (FDP-KRB, and FDP-SD), as well as when controlling the FDR using TDC.
The FDR/FDP threshold is $\alp\in\{0.1,0.2\}$ and the confidence parameters for FDP control is $\gam\in\{0.1, 0.05\}$.
\label{fig:lin_reg}}
\end{SCfigure}

We also examined, using real data, the performance of FDP-SD in the peptide detection problem.
Specifically, we used the same methodology as described in \cite{emery:multipleRECOMB} --- recapped here in \supsec~\ref{supsec:isb18_analysis} ---
for detecting peptides in the ISB18 data set~\cite{klimek:standard}.
This process generated 900 sets of paired target and decoy scores assigned to each peptide in our database.
We then applied TDC and FDP-SD to each of these 900 sets using
an FDR/FDP threshold of $\alpha$ = 5\%, and confidence level of $100(1-\gam)$ = 95\%.
We relied on the controlled nature of the experiment that generated the ISB18 data to estimate the FDP in each case
(\supsec~\ref{supsec:isb18_analysis}).

Even though our model is just an approximation of the real peptide detection problem, FDP-SD's FDP exceeded $\alp$
in only 36/900, or 4\% of the samples, which is less than the allowed error rate of $\gamma=0.05$.
Additionally, Figure~\ref{fig:rel_loss_isb18}
shows how the relative power loss associated with using FDP-SD is distributed across the 900 samples (median power loss is 6.7\%).
For reference, the distribution of TDC's FDP in this experiment is given in Figure~\ref{fig:intro} (left).

\begin{SCfigure}
\includegraphics[width=3in]{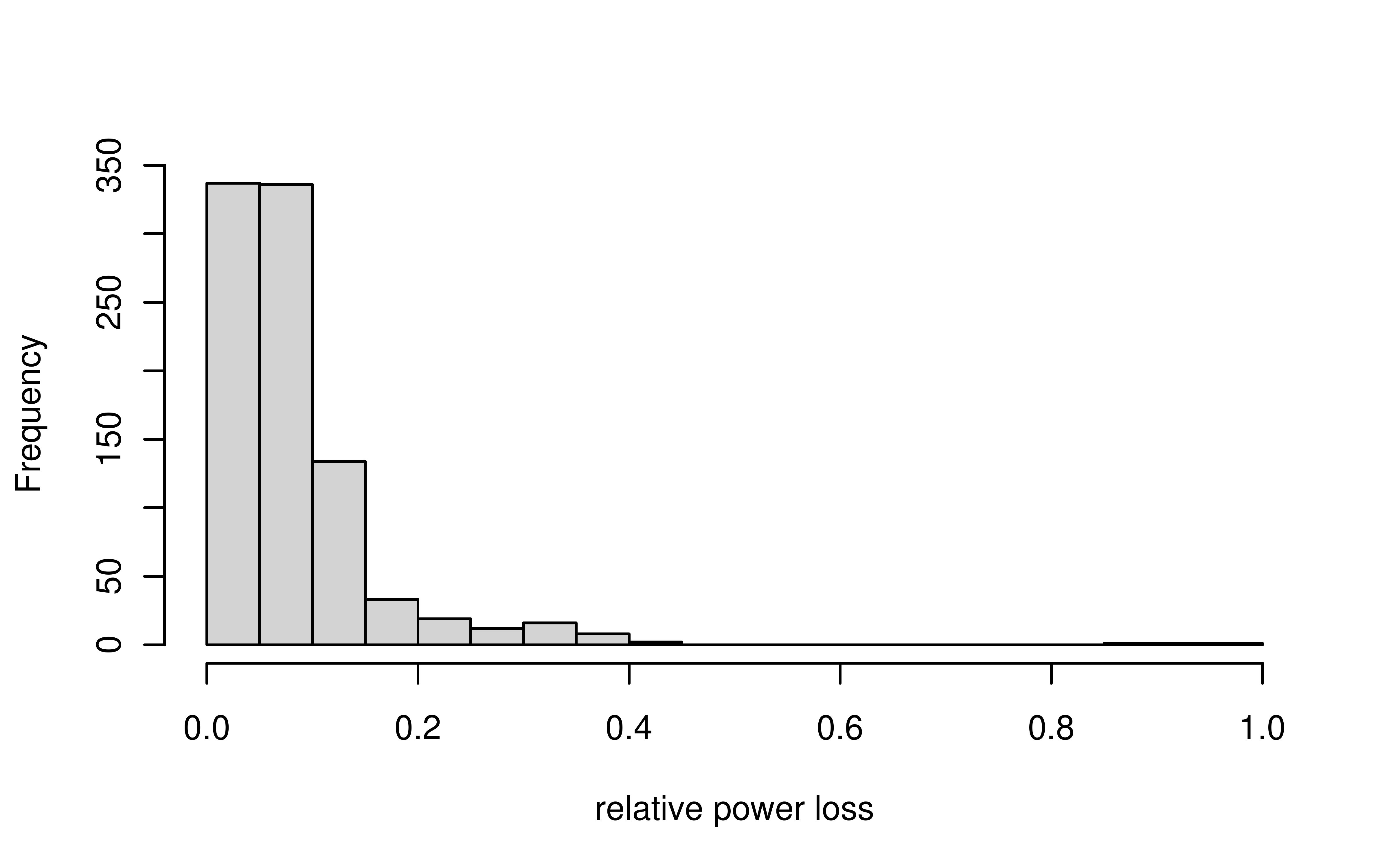}
\caption{\textbf{Relative power loss in detected peptides when controlling the FDP instead of the FDR (ISB18).}
The distribution of the relative loss in power (number of correctly detected peptides) when using FDP-SD to control the FDP
at level $\alp=0.05$ with 95\% confidence compared with using TDC 
to control the FDR at the same $\alp$. The median power loss among 900 samples was 6.7\%.
\label{fig:rel_loss_isb18}}
\end{SCfigure}

Finally, it is instructive to look at what happens in the spectrum-ID problem when we vary $m$ while keeping the other parameters constant
($\alp=0.05,\gam=0.05,\pi_0=0.5$).
\supfig~\ref{supfig:vary_m} shows that, as expected, increasing $m$ yields diminished variability in TDC's FDP  (top row). At the same time the power loss
associated with FDP-SD's increased confidence also diminishes (middle row).
A similar evolution is observed in \supfig~\ref{supfig:vary_pi0} as we decrease $\pi_0$ while keeping all other parameters the same ($\alp=0.05,\gam=0.05,m=$2K).
This is not surprising because increasing $m$ and decreasing $\pi_0$ have the same effect of increasing the number of discoveries.

%
%

%
%

\section{Discussion}

FDP-SD was developed to address the gap between controlling the FDR and the FDP in a competition-based setup. 
In practice, this difference can be substantial, particularly when the list of discoveries is not very large.
Our procedure was developed independently of the recent work of \KR~\cite{katsevich:simultaneous}. The latter
provides a much more general framework that can be applied to produce an alternative to FDP-SD, but as we show
here, our more focused approach provides a non-trivial advantage.

Another related work is by Janson and Su who, while focusing on $k$-FWER control (i.e., no more than $k$ false discoveries),
suggest how one can use their approach to gain control of the FDP (FDX-control)~\cite{janson:familywise}.
However, two of their suggestions are computationally impractical while the third is based on the
Romano-Wolf heuristic~\cite{romano:control} rather than rigorously proved. Interestingly, we believe that a simple variation on that heuristic yields
FDP-SD, which we propose and rigorously establish the validity of (Section \ref{sec:FDP-SD}).

Complexity-wise, FDP-SD requires sorted data, but beyond that it is linear; hence, its runtime complexity is $O(m\log m)$.

In terms of future research there are a couple of avenues we would like to explore. First,
we showed how to extend FDP-SD so that it can control the FDP while taking advantage of multiple decoys using a pre-determined choice of $c$ and $\lambda$.
However, as shown by Emery et al.~in the context of FDR control the choice of $c$ and $\lambda$ can greatly affect the power of the procedure.
This suggests we can similarly benefit from such optimization when controlling the FDP.

Second, we provided a somewhat contrived procedure that outperforms FDP-SD while still controlling the FDP by improving on the latter in a very specific scenario.
We would like to explore whether FDP-SD can be improved upon in a more systematic way including considering the general approach to such questions that was
recently proposed by Goeman et al.~\cite{goeman:only}.


\section{Supplementary Material}

\subsection{Notations and Abbreviations}

\begin{table}[h!]
\fontsize{8.5}{10}\selectfont
\centering
\begin{tabular}{cp{5in}}
\hline 
\textbf{Variable} & \textbf{Definition} \\
\hline
$m$ & the number of hypotheses (e.g., PSMs, features, peptides)\\
$\alp$ & the FDR/FDP threshold\\
$\gam$ & $1-\gam$ is the confidence level\\
$N$ & the set of indices of the true null hypotheses (unobserved, could be a random set)\\
$\sig_i$ & a (virtual) spectrum\\
$X_i$ & the score of the match between $\sig_i$ and its ``generating peptide''\\
$Y_i$ & the score of the best match to $\sig_i$ in the target database minus the generating peptide (if it exists)\\
$Z_i$ & the target score (observed, the higher the score the less likely $H_i$ is, $\max\{X_i,Y_i$\} in simulated spectrumID)\\
$\til Z_{i}$ & decoy/knockoff score (generated by the user, the score of the best match to $\sig_i$ in the decoy database in simulated spectrumID)\\
$L_i$ & with values in $\{-1,1\}$ the target/decoy win labels (assigned, ties are randomly broken or the corresponding hypotheses are dropped)\\
$W_i$ & the winning score (assigned, WLOG assumed in decreasing order)\\
$D_i$ & the number of decoy wins in the top $i$ scores\\
$T_i$ & the number of target wins in the top $i$ scores\\
$N^t_i$ & the number of true null target wins in the top $i$ scores\\
$T'_i$ & the number of true/correct discoveries in the top $i$ scores ($T'_i=T_i-N^t_i$)\\
$Q_i$ & the FDP among the target wins in the top $i$ scores\\
$V_d$ & the number of true null target wins before the $d$th decoy win.\\
\hline 
\end{tabular}
\caption{\textbf{Commonly used notations and their definitions.}
  \label{suptable:definitions}}
\end{table}

\begin{table}[h!]
\fontsize{8.5}{10}\selectfont
\centering
\begin{tabular}{lp{5in}}
\hline 
\textbf{Abbreviation} & \textbf{Definition} \\
\hline
MS/MS & Tandem Mass Spectrometry\\
PSM & Peptide-Spectrum Match (the match between a spectrum and its best matching database peptide)\\
spectrum-ID & Spectrum Identification (the problem of matching spectra to the peptides that generated them)\\
RV & Random Variable\\
FDP & False Discovery Proportion (the proportion of the discoveries which is false - a RV)\\
FDX & False Discovery Exceedance (an alternative term for FDP-control that is used in the literature)\\
FDR & False Discovery Rate (the expected value of the FDP taken with respect to the true nulls)\\
TDC & Target Decoy Competition (canonical approach to FDR control)\\
FDP-SD & FDP-Stepdown (our recommended new procedure to control the FDP)\\
FDP-KRB & FDP-\KR\ Band (an alternative new FDP-controlling procedure based on the \KR\ band)\\
GWAS & Genome-Wide Association Studies (here referring to a specific analysis of 8 traits using Biobank data)\\
\hline 
\end{tabular}
\caption{\textbf{Commonly used abbreviations/names and their definitions.}
  \label{suptable:abbreviations}}
\end{table}

\clearpage

\subsection{Brief background on shotgun proteomics and the spectrum-ID model}\label{sec:background}

Tandem mass spectrometry (MS/MS) currently provides the most efficient means of studying proteins in a high-throughput fashion.
\recomb{} 
{As such, MS/MS is the driving technology for much of the rapidly growing field of proteomics --- the large scale study of proteins.
Proteins are the primary functional molecules in living cells, and knowledge of the protein complement in a cellular population provides
insight into the functional state of the cells.
Thus, MS/MS can be used to functionally characterize cell types, differentiation stages, disease states, or species-specific differences.

}In a ``shotgun proteomics'' MS/MS experiment, the proteins that are extracted from a complex biological sample 
are not measured directly. For technical reasons, the proteins are first digested into shorter
chains of amino acids called ``peptides.''  The peptides are then run through the mass spectrometer, in which distinct peptide sequences generate
corresponding spectra. A typical 30-minute MS/MS experiment will generate approximately 18,000 such spectra.
Canonically, each observed spectrum is generated by a single peptide.
Thus, the first goals of the downstream analysis are to identify which peptide generated each of the observed spectra
(the spectrum-ID problem mentioned above) and to determine which peptides and which proteins were present in the sample
(the peptide/protein detection problems).

In each of those three problems the canonical approach to determine the list of discoveries is by controlling the FDR
through some form of target-decoy competition.
One reason this approach was adopted, rather than relying on standard methods for control of the FDR such as the
procedures by Benjamini and Hochberg~\cite{benjamini:controlling} or Storey~\cite{storey:direct}, is that the latter require
sufficiently informative p-values and, initially, no such p-values were computed in this context (using the decoys we can always assign
a ``1-bit p-value'' to the hypotheses but those are not informative enough to obtain effective results using the latter procedures).
Moreover, the proteomics dataset typically consist of both ``native'' spectra
(those for which their generating peptide is in the target database) and ``foreign'' spectra
(those for which it is not). These two types of spectra create different types of false positives, implying
that we typically cannot apply the standard FDR controlling procedures to the spectrum-ID problem
even if we are able to compute p-values~\cite{keich:controlling}.

A simple model that captures the distinction between native and foreign spectra
and which here we refer to as ``the spectrum-ID model'' is described in~\cite{keich:improved,keich:controlling}.
Briefly, each virtual ``spectrum'' $\sig_i$ is associated
with three randomly drawn scores: the score $X_i$ of the match between $\sig_i$ and its generating peptide, the score $Y_i$ 
of the best match to $\sig_i$ in the target database minus the generating peptide (if $\sigma_i$ is native),
and the score $\til Z_i$ of the best match to $\sig_i$ in the decoy database. The three scores are drawn
independently of one another as well as of the corresponding scores of all other spectra.
More specifically, $Y_i$ and $\til Z_i$ are sampled from a null distribution (which can be spectrum-specific),
whereas $X_i$ is sampled from an alternative distribution for a native $\sig_i$, and $X_i$ is set to $-\infty$
for a foreign $\sig_i$. The target PSM score is $Z_i=X_i\vee Y_i$, where $x\vee y$ denotes $\max\{x,y\}$, and the decoy PSM score is $\til Z_i$.
Finally, the PSM is incorrect when $Y_i\vee\til Z_i>X_i$.

Notably, conditional on the scores $X_i$, this model satisfies Assumption \ref{fund_assump} from the main paper:
conditional on the PSM being incorrect ($Y_i\vee\til Z_i>X_i$) it is easy to see that
$P(Y_i>\til Z_i)=P(Y_i<\til Z_i)$ independently of everything else.
That said, it is worth pointing out a couple of features
that are distinct to this setup. First, the set $N$ of true null hypotheses is random because it depends on the decoy
scores $\til Z_i$ (as well as on the random target scores $Y_i$): by definition a PSM is incorrect if $Y_i\vee\til Z_i>X_i$.
Second, a false null (correct PSM) has to correspond to a target win. 
This is not the case in general. For example, in the feature selection problem, a feature is a false null when its coefficient in
the regression model is not zero.  It is possible for such a feature to have a lower score than its corresponding
knockoff and hence to be counted as a decoy win.

\subsection{Procedures in Algorithmic Format}
\label{supsec:algs}

\vspace*{\fill}

\begin{algorithm}
\textbf{\caption{TDC}\label{algorithm:TDC}}
\SetKwFor{For}{For}{let:}{endfor}
\SetKwIF{If}{ElseIf}{Else}{If}{then:}{else if:}{else:}{endif} 
\KwIn{
    \begin{alglist}
        \item an FDR threshold $\alpha$\;
        \item a list of labels $L_i=\pm1$ where $1$ indicates a target win and $-1$ a decoy win (sorted so that the corresponding scores $W_i$ are in decreasing order: $W_1\ge W_2\ge\dots\ge W_m$)\;
    \end{alglist}
}
\KwOut{
    \begin{alglist}
        \vspace{-0.35em}
        \item an index $k_{\text{TDC}}$ specifying that target wins in the top $k_{\text{TDC}}$ hypotheses are discoveries\;
    \end{alglist}
}
\vspace{-0.5em}
\dottedhfill\\
\setstretch{1.4}
\For{$i=0$ \KwTo $m$} {
    $D_i$ be the number of $-1$'s in $\{L_1,\ldots,L_i\}$\;
    $T_i := i - D_i$\;
}
$M := \{k \in \{1, \ldots, m\} : (D_k+1)/T_k \leq \alpha\}$\;
\eIf{$M = \emptyset$}{
    \Return $k_{\text{TDC}} := 0$\;
}{
    \Return $k_{\text{TDC}} := \max(M)$\;
}
\end{algorithm}

\vspace*{\fill}

\newpage

\vspace*{\fill}

\begin{algorithm}
\textbf{\caption{FDP-KRB}\label{algorithm:FDP-KRB}}
\SetKwFor{For}{For}{let:}{endfor}
\SetKwFor{While}{While}{let:}{endw}
\SetKwIF{If}{ElseIf}{Else}{If}{then:}{else if:}{else:}{endif}
\SetKwInput{KwData}{Initialization}
\KwIn{
    \begin{alglist}
        \item an FDP threshold $\alpha$\;
        \item a confidence parameter $\gam$ (for a $1-\gam$ confidence level)\;
        \item a list of labels $L_i=\pm1$ where $1$ indicates a target win and $-1$ a decoy win (sorted so that the corresponding scores $W_i$ are in decreasing order: $W_1\ge W_2\ge\dots\ge W_m$)\;
    \end{alglist}
}
\KwOut{
    \begin{alglist}
        \item an index $k_{\text{KR}}$ specifying that target wins in the top $k_{\text{KR}}$ hypotheses are discoveries\;
    \end{alglist}
}
\vspace{-0.5em}
\dottedhfill\\
\setstretch{1.2}
$C \coloneqq -\log(\gamma)/\log(2-\gamma)$\;
\For{$i=1$ \KwTo $m$} {
	$D_i$ be the number of $-1$'s in $\{L_1,\ldots,L_i\}$\;
    $T_i := i - D_i$\;
}
$M := \{k \in \{1, \ldots, m\} : [C(D_k+1)]/T_k \leq \alpha\}$\;
\eIf{$M = \emptyset$}{
    \Return $k_{\text{KR}} := 0$\;
}{
    \Return $k_{\text{KR}} := \max(M)$\;
}
\end{algorithm}


\vspace*{\fill}

\newpage

\vspace*{\fill}

\begin{algorithm}
\textbf{\caption{FDP-SD}\label{algorithm:FDP-SD}}
\SetKwFor{For}{For}{let:}{endfor}
\SetKwIF{If}{ElseIf}{Else}{If}{then:}{else if:}{else:}{endif} 
\KwIn{
    \begin{alglist}
        \item an FDP threshold $\alpha$\;
        \item a confidence parameter $\gamma$ (for a $1 - \gamma$ confidence level)\;
        \item a list of labels $L_i=\pm1$ where $1$ indicates a target win and $-1$ a decoy win (sorted so that the corresponding scores $W_i$ are in decreasing order: $W_1\ge W_2\ge\dots\ge W_m$)\;
    \end{alglist}
}
\KwOut{
    \begin{alglist}
        \vspace{-0.35em}
        \item an index $k_{\text{FDP-SD}}$ specifying that target wins in the top $k_{\text{FDP-SD}}$ hypotheses are discoveries\;
    \end{alglist}
}
\vspace{-0.5em}
\dottedhfill\\
\setstretch{1.4}
$i_0 := \max\{1, \lceil\left(\lceil\log_2\left(1/\gam\right)\rceil-1\right)/\alp\rceil\}$\;
\For{$i=i_0$ \KwTo $m$} {
    $D_i$ be the number of $-1$'s in $\{L_1,\ldots,L_i\}$\;
    $\delta_i := \max\left\{d\in\{0,1,\dots,i\}\,:\,F_{B(\lfloor (i-d)\alpha \rfloor + 1 + d,1/2)}(d) \le \gam  \right\}$ where $F_{B(n,p)}$ denotes the CDF of a binomial $B(n,p)$ RV\;
}
\eIf{$D_{i_0} \leq \delta_{i_0}$}{
    \Return $k_{\text{FDP-SD}} := \max\left\{i \in \{i_0, \ldots, m\}\,:\,D_j\le\del_j\text{ for all }j=i_0,i_0+1,\dots,i\right\}$\;
}{
    \Return $k_{\text{FDP-SD}} := 0$\;
}
\end{algorithm}

\vspace*{\fill}

\newpage

\begin{algorithm}[h]
\footnotesize
\textbf{\caption{FDP-SD (randomized)}\label{algorithm:r-FDP-SD}}
\SetKwFor{For}{For}{let:}{endfor}
\SetKwFor{While}{While}{let:}{endw}
\SetKwIF{If}{ElseIf}{Else}{If}{then:}{else if:}{else:}{endif}
\SetKwInput{KwData}{Initialization}
\KwIn{
    \begin{alglist}
        \item an FDP threshold $\alpha$\;
        \item a confidence parameter $\gamma$ (for a $1 - \gamma$ confidence level)\;
        \item a list of labels $L_i=\pm1$ where $1$ indicates a target win and $-1$ a decoy win (sorted so that the corresponding scores $W_i$ are in decreasing order: $W_1\ge W_2\ge\dots\ge W_m$)\;
    \end{alglist}
}
\KwOut{
    \begin{alglist}
        \vspace{-0.35em}
        \item an index $k_{\text{r-FDP-SD}}$ specifying that target wins in the top $k_{\text{r-FDP-SD}}$ hypotheses are discoveries\;
    \end{alglist}
}
\vspace{-0.5em}
\dottedhfill\\
\setstretch{1.2}
$i_0 := \max\{1, \lceil\left(\lceil\log_2\left(1/\gam\right)\rceil-1\right)/\alp\rceil\}$\;
\For{$i=i_0$ \KwTo $m$} {
    $D_i$ be the number of $-1$'s in $\{L_1,\ldots,L_i\}$\;
    $\delta_i := \max\left\{d\in\{0,1,\dots,i\}\,:\,F_{B(\lfloor(i-d)\alpha\rfloor + 1 + d,1/2)}(d) \le \gam  \right\}$ where $F_{B(n,p)}$ denotes the CDF of a binomial $B(n,p)$ RV\;
}
\textbf{Set} $i := i_0$ and $\delta_{i_0 - 1} := -1$ and $\bar\delta_{i_0 - 1} = 0$\;
\While{$i\leq m$}{
    $k_0 := \lfloor(i-\delta_i)\cdot\alpha\rfloor + 1$\;
    $k_1 := \lfloor((i - (\delta_i + 1))\cdot \alpha\rfloor + 1$\;
    $p_0 := F_{B(k_0+\delta_i, 1/2)}(\delta_i)$\;
    $p_1 := F_{B(k_1+\delta_i+1, 1/2)}(\delta_i+1)$\;
    $w_i := (p_1 - \gam)/(p_1 - p_0)$\;
    \eIf{$\bar\delta_{i-1} = \delta_i + 1$}{
            $\bar\delta_i \eqq \bar\delta_{i-1}$\;
        }{
			\eIf{$\delta_i > \delta_{i-1}$}{
				$w' \eqq w_i$
			}{
				$w' \eqq w_i/w_{i-1}$\;
			}
        }
    \textbf{Randomly set} $\bar\delta_i \eqq \delta_i$ or $\bar\delta_i \eqq \delta_i + 1$ with probabilities $w'$ and $1 - w'$ respectively\;
    \eIf{$D_i \leq \bar\delta_i$}{
        $i\mapsto i+1$\;
    }{
        {\bf break}\;
    }
}
\eIf{$D_{i_0} \leq \bar\delta_{i_0}$}{
    \Return $k_{\text{r-FDP-SD}} := i-1$\;
}{
    \Return $k_{\text{r-FDP-SD}} := 0$\;
}
\end{algorithm}

\newpage

\begin{algorithm}
\textbf{\caption{FDP-SD (multiple decoys)}\label{algorithm:FDP-SDm}}
\SetKwFor{For}{For}{let:}{endfor}
\SetKwFor{While}{While}{let:}{endw}
\SetKwIF{If}{ElseIf}{Else}{If}{then:}{else if:}{else:}{endif}
\SetKwInput{KwData}{Initialization}
\KwIn{
    \begin{alglist}
        \item an FDP threshold $\alpha$\;
        \item a confidence parameter $\gam$ (for a $1-\gam$ confidence level)\;
        \item $d_0$ competing decoys\;
        \item competition parameters $c = i_c/(d_0+1)$ and $\lambda = i_\lambda/(d_0+1)$ for $i_c, i_\lambda \in \{1, \ldots, d_0\}$\;
        \item a list of labels $L_i \in \{-1, 0, 1\}$ where $1$ indicates a target win, $-1$ a decoy win and $0$ an uncounted hypothesis (sorted so that the corresponding scores $W_i$ are decreasing: $W_1\ge W_2\ge\dots\ge W_m$)\;
    \end{alglist}
}
\KwOut{
    \begin{alglist}
        \item an index $k_{\text{FDP-SDm}}$ specifying that target wins in the top $k_{\text{FDP-SDm}}$ hypotheses are discoveries\;
    \end{alglist}
}
\vspace{-0.5em}
\dottedhfill\\
\setstretch{1.4}
$R := (1 - \lambda) / (c + 1 - \lambda)$\;
$i_0 := \max\{1, \lceil \left(\lceil \log_{1-R}(\gamma) \rceil - 1 \right) / \alp \rceil\}$\;
\For{$i=i_0$ \KwTo $m$} {
	$D_i$ be the number of $-1$'s in $\{L_1,\ldots,L_i\}$\;
	$\delta_i \coloneqq \max\left\{d\in\{0,1,\dots,i\}\,:\,F_{B(\lfloor(i-d)\alpha\rfloor + 1 + d,R)}(d) \le \gam  \right\}$ where $F_{B(n,p)}$ denotes the CDF of a binomial $B(n,p)$ RV\;
}
\eIf{$D_{i_0} \leq \delta_{i_0}$}{
    \Return $k_{\text{FDP-SDm}} := \max\left\{i \in \{i_0, \ldots, m\}\,:\,D_j\le\del_j\text{ for all }j=i_0,i_0+1,\dots,i\right\}$\;
}{
    \Return $k_{\text{FDP-SDm}} := 0$\;
}
\end{algorithm}

\newpage

\begin{algorithm}[h]
\footnotesize
\textbf{\caption{FDP-SD (multiple decoys, randomized)}\label{algorithm:r-FDP-SDm}}
\SetKwFor{For}{For}{let:}{endfor}
\SetKwFor{While}{While}{let:}{endw}
\SetKwIF{If}{ElseIf}{Else}{If}{then:}{else if:}{else:}{endif}
\SetKwInput{KwData}{Initialization}
\KwIn{
    \begin{alglist}
        \item an FDP threshold $\alpha$\;
        \item a confidence parameter $\gam$ (for a $1-\gam$ confidence level)\;
        \item $d_0$ competing decoys\;
        \item competition parameters $c = i_c/(d_0+1)$ and $\lambda = i_\lambda/(d_0+1)$ for $i_c, i_\lambda \in \{1, \ldots, d_0\}$\;
        \item a list of labels $L_i \in \{-1, 0, 1\}$ where $1$ indicates a target win, $-1$ a decoy win and $0$ an uncounted hypothesis (sorted so that the corresponding scores $W_i$ are decreasing: $W_1\ge W_2\ge\dots\ge W_m$)\;
    \end{alglist}
}
\KwOut{
    \begin{alglist}
        \item an index $k_{\text{r-FDP-SDm}}$ specifying that target wins in the top $k_{\text{r-FDP-SDm}}$ hypotheses are discoveries\;
    \end{alglist}
}
\vspace{-0.5em}
\dottedhfill\\
\setstretch{1.2}
$R := (1 - \lambda) / (c + 1 - \lambda)$\;
$i_0 := \max\{1, \lceil \left(\lceil \log_{1-R}(\gamma) \rceil - 1 \right) / \alp \rceil\}$\;
\For{$i=i_0$ \KwTo $m$} {
    $D_i$ be the number of $-1$'s in $\{L_1,\ldots,L_i\}$\;
    $\delta_i := \max\left\{d\in\{0,1,\dots,i\}\,:\,F_{B(\lfloor(i-d)\alpha\rfloor + 1 + d,R)}(d) \le \gam  \right\}$ where $F_{B(n,p)}$ denotes the CDF of a binomial $B(n,p)$ RV\;
}
\textbf{Set} $i := i_0$ and $\delta_{i_0 - 1} := -1$ and $\bar\delta_{i_0 - 1} = 0$\;
\While{$i\leq m$}{
    $k_0 := \lfloor(i-\delta_i)\cdot\alpha\rfloor + 1$\;
    $k_1 := \lfloor((i - (\delta_i + 1))\cdot \alpha\rfloor + 1$\;
    $p_0 := F_{B(k_0+\delta_i, R)}(\delta_i)$\;
    $p_1 := F_{B(k_1+\delta_i+1, R)}(\delta_i+1)$\;
    $w_i := (p_1 - \gam)/(p_1 - p_0)$\;
    \eIf{$\bar\delta_{i-1} = \delta_i + 1$}{
            $\bar\delta_i \eqq \bar\delta_{i-1}$\;
        }{
			\eIf{$\delta_i > \delta_{i-1}$}{
				$w' \eqq w_i$
			}{
				$w' \eqq w_i/w_{i-1}$\;
			}
        }
    \textbf{Randomly set} $\bar\delta_i \eqq \delta_i$ or $\bar\delta_i \eqq \delta_i + 1$ with probabilities $w'$ and $1 - w'$ respectively\;
    \eIf{$D_i \leq \bar\delta_i$}{
        $i\mapsto i+1$\;
    }{
        {\bf break}\;
    }
}
\eIf{$D_{i_0} \leq \bar\delta_{i_0}$}{
    \Return $k_{\text{r-FDP-SDm}} := i-1$\;
}{
    \Return $k_{\text{r-FDP-SDm}} := 0$\;
}
\end{algorithm}

\clearpage

\subsection{Proof of Theorem \ref{thm:FDP-SD}}
\label{supsec:proof_thm_FDP-SD}

\begin{proof}
	
To simplify notation, let $\tau=k_{\text{FDP-SD}}$ and $\delta_i = \delta_{\alpha, \gam}(i)$.
Denote the number of target-winning false nulls (i.e., correct target discoveries) among the top $i$ scores by
$$A_i^t := T_i-N_i^t,$$ 
and the number of those which are decoy-winning by
$$A_i^d:=\sum_{j=1}^i \mathbbm{1}_{\{L_j = -1,j\notin N\}}.$$
Let $A_i := A_i^d + A_i^t$ be the total number of false nulls among the top $i$ scores and $I_A:=\{1,\dots,m\}\setminus N$ be the set of non-null indices.

By the law of total probability, it is enough to prove that FDP-SD controls the FDP for a fixed collection of winning scores and fixed positions and labels of the false nulls. Thus, assume that it is given $W_i$ (in decreasing order: $W_1 \geq \dots \geq W_m$) and $\{L_i : i \in I_A\}$. Note that, by Assumption \ref{fund_assump}, given such information, the true null labels $\{L_i\,:\,i\in N\}$ are i.i.d. uniform $\pm1$ RVs.

Let 
$$
M := \{i\in\mathbb{N}:i_0 \leq i\leq m \mbox{ and } i-A_i^t-\delta_i>\alpha(i-\delta_i)\}.
$$
If $M \neq \emptyset$, define $j := \min M$; otherwise, set $j:=\infty$.
Since the labels and positions of the false nulls are fixed, we have that $A_i^t$ and $A_i^d$ are fixed, and therefore, so too is $j$.
We first consider the case where $j$ is finite.
\begin{Lemma}
\label{lemmma1_for_SD}
Let $Q_\tau$ be the FDP in the list of discoveries resulting from FDP-SD. If $Q_\tau > \alpha$ and $j < \infty$ then $\tau\geq j$.
\end{Lemma}
\begin{proof}
If $Q_\tau>\alpha$, then $\tau \geq i_0 > 0$. As the procedure ended on index $\tau$, either $\tau = m \ge j$ and the conclusion follows,
or $\tau<m$ and $D_{\tau+1}>\delta_{\tau+1}$ whilst $D_\tau\le\del_\tau$. Since $\del_i\le\del_{i+1}$, it follows that
$D_{\tau}\leq \delta_\tau \leq \delta_{\tau + 1}<D_{\tau+1}$.
But all terms in this string of inequalities are integers and $0\leq D_{\tau+1}-D_\tau \leq 1$. Therefore, $D_\tau = \delta_\tau = \delta_{\tau + 1}$.
In particular, if $Q_\tau>\alpha$ then
\begin{equation*}
        \alpha < Q_\tau = \frac{\tau - A_\tau^t-D_\tau}{\tau-D_\tau}
         = \frac{\tau - A_\tau^t-\delta_\tau}{\tau-\delta_\tau} ,
\end{equation*}
giving $\tau\in M$, and thus, $\tau\ge j=\min M$.
\end{proof}
\begin{Remark}
	\label{rem:tau_lt_m}
We note from the proof that if $\tau<m$ and $Q_\tau > \alpha$ then
$\tau\in M \neq \emptyset$. In particular, if $j = \infty$ and $Q_\tau > \alpha$ then it must be the case that $\tau = m$.
\end{Remark}
Assuming that $j<\infty$, it follows that $\{Q_\tau >\alpha\}\subseteq \{\tau\geq j\}$, and since
\[
\{\tau\geq j\} = \bigcap_{i = i_0}^j  \{D_i \leq \delta_i\} \subseteq \{D_j \leq \delta_j\},
\]
it suffices to show $P(D_j \leq \delta_j)\leq \gamma$.

Still assuming $j<\infty$, $(j-\delta_j)\alpha<j-A_j^t-\delta_j$ by definition. In particular,
\begin{equation*}
           k(\delta_j)  := \lfloor (j-\delta_j)\alpha \rfloor + 1
        \leq j - A_j^t - \delta_j,
\end{equation*}
and therefore, $k(\delta_j)+\delta_j\leq j - A_j^t$.

By definition of $\delta_j$,
$$
F_{B(k(\delta_j)+\delta_j,\frac{1}{2})}(\delta_j) \leq \gamma.
$$
The CDF of a $\operatorname{Binomial}(n,p)$ decreases with $n$, so the last two inequalities imply that
\begin{equation}
\label{eq_sd1}
F_{B(j - A_j^t,\frac{1}{2})}(\delta_j)
    \leq 
    F_{B(k(\delta_j)+\delta_j,\frac{1}{2})}(\delta_j)
    \leq \gamma.
\end{equation}

Note that the number of true nulls among the top $j$ hypotheses is the fixed quantity $j - A_j$. Recall, by Assumption \ref{fund_assump}, the labels of those true nulls are i.i.d. uniform $\pm1$ RVs.
Hence, $N^d_j$, the number of decoy-winning true nulls in the top $j$ scores, follows a $\operatorname{Binomial}(j-A_j,\frac{1}{2})$ distribution, and thus,
\begin{equation}
	\label{eq:adhoc1}
    P(D_j\leq \delta_j) = P(N^d_j+A_j^d \leq \delta_j) = P(N^d_j \leq \delta_j - A_j^d) = F_{B(j-A_j,\frac{1}{2})}(\delta_j - A_j^d).
\end{equation}
\begin{Remark}
	\label{rem:binom}
Note that for $k,l,m,n\in\mathbb{N}$ with $l\geq n$,
$$F_{B(m,p)}(k)\leq F_{B(n+m,p)}(l+k).$$
Indeed, if there are $\le k$ successes in the first $m$ trials then there will be $\le n+k \le l+k$ successes in all $n+m$ trials.
\end{Remark}
From Lemma \ref{lemmma1_for_SD}, \eqref{eq_sd1}, \eqref{eq:adhoc1} and the last remark, we conclude that
\begin{align*}
    P(Q_\tau >\alpha) &\le P(D_j\leq \delta_j) \\
& =  F_{B(j-A_j,\frac{1}{2})}(\delta_j - A_j^d)\\
& =  F_{B(j - A_j^t - A_j^d,\frac{1}{2})}(\delta_j - A_j^d)\\
& \leq  F_{B(j - A_j^t,\frac{1}{2})}(\delta_j)\\
& \leq \gamma,
\end{align*}
thus establishing that, when $j < \infty$, FDP-SD controls the FDP with confidence $1 - \gamma$.

Next, consider the case when $M=\emptyset$ (equivalently, $j = \infty$). As noted in Remark \ref{rem:tau_lt_m}, if $\tau < m$ and $Q_\tau > \alpha$,
then $\tau \in M$, resulting in a contradiction. Therefore, when $M=\emptyset$ and $Q_\tau > \alpha$, $\tau$ must equal $m$,
which in turn implies that $D_m \leq \delta_m$. Hence,
$$
\{Q_\tau>\alpha, M=\emptyset\} = \{D_m \leq \delta_m, Q_\tau>\alpha, \tau=m, M=\emptyset\} \subseteq  \{D_m \leq \delta_m, Q_m>\alpha, M=\emptyset\}.
$$

Since
$$
Q_m = \frac{m-A_m^t-D_m}{m-D_m},
$$
and $Q_m>\alpha$, it follows that $(m-D_m)\alpha<m-A_m^t-D_m$ and
$$
\lfloor (m-D_m)\alpha\rfloor + 1\leq m-A_m^t-D_m.
$$
Therefore, assuming $M=\emptyset$,
\begin{equation}
	\label{eq:adhoc2}
\begin{aligned}
    P(Q_\tau > \alp)
    \le\,   & P(D_m\leq\delta_m,\, Q_m > \alpha)  \\
    \le\,  & P\left[D_m\leq\delta_m,\, \lfloor (m-D_m)\alpha\rfloor + 1 \leq m - A_m^t - D_m\right].
\end{aligned}
\end{equation}
With
$$k(m,d) := \lfloor(m-d)\alpha\rfloor + 1,$$
denote
$$p(m,d) := F_{B(k(m,d)+d,\frac{1}{2})}(d).$$
To express the event $D_m\le\del_m$ in terms of $p(m,d)$, the following lemma is necessary.

\begin{Lemma}
	\label{Lem:pmd_is_mono}
Let $d$ and $d'$ be integers such that $0\leq d \leq d' \leq m$. Then, $p(m,d) \leq p(m,d')$.
\end{Lemma}
\begin{proof}
Note that for $x,y\in\R$, $\lfloor y\rfloor -\lfloor x\rfloor \le \lceil y-x\rceil$, hence for $d\leq d'$
\begin{align*}
    0 \,\le \,k(m,d) - k(m,d')
    & = 
    \lfloor(m-d)\alpha\rfloor -\lfloor(m-d')\alpha\rfloor\\
    & \leq 
    \lceil(m-d)\alpha -(m-d')\alpha\rceil\\
    & =
    \lceil(d'-d)\alpha\rceil\\
    & \leq
    d' - d.
\end{align*}
Therefore,
$$
0
\leq
[k(m,d')+d'] - [k(m,d)+d]
\leq
d' - d.
$$
It follows from Remark \ref{rem:binom} that
\begin{align*}
    p(m,d)
    & = 
    F_{B(k(m,d)+d,\frac{1}{2})}(d) \\
    & \leq
    F_{B(k(m,d')+d',\frac{1}{2})}(d')\\
    & = 
    p(m,d').
\end{align*}
\end{proof}
\begin{Corollary}
For $d\in\N$ with $d\le m$,
$p(m,d)\leq p(m,\delta_m)$ if and only if $d \leq \delta_m$, if and only if $p(m,d)\leq \gamma$.
\end{Corollary}
\begin{proof}
The equivalences follow immediately from the last lemma and the fact that, by definition,
$\delta_m = \max\{d\in\{-1,0,\dots,m\}\,:\, p(m,d) \leq \gamma\}.$
\end{proof}

Thus, continuing from \eqref{eq:adhoc2},
$$
P(Q_\tau > \alpha)
\leq
P\left[p(m,D_m)\leq \gamma\,,\, k(m,D_m)+D_m\leq m- A_m^t\right].
$$
Note that for $n,d\in\N$ such that $n\geq k(m,d)+d$, 
$$
p(m,d)  = F_{B(k(m,d)+d,\frac{1}{2})}(d) \ge F_{B(n,\frac{1}{2})}(d).
$$
So in this case,
$$
\{p(m,d)\le \gamma \} \subseteq \{F_{B(n,\frac{1}{2})}(d)\le \gamma \}.
$$
It follows that with $n = m - A_m^t$ and $d=D_m$
\begin{align*}
    P(Q_\tau > \alpha)
    & \leq P\big[p(m,D_m) \le \gamma \,,\, k(m,D_m)+D_m\leq m-A_m^t\big]\\
    & \leq P\big[F_{B(m - A_m^t,\frac{1}{2})}(D_m) \le \gamma \,,\, k(m,D_m)+D_m\leq m-A_m^t\big]\\
    & \leq P\big[F_{B(m - A_m^t,\frac{1}{2})}(D_m) \le \gamma \big].
\end{align*}
By Remark \ref{rem:binom},
$$
F_{B(m-A_m^t-A_m^d,\frac{1}{2})}(D_m-A_m^d) \leq F_{B(m-A_m^t,\frac{1}{2})}(D_m).
$$
Therefore,
$$
P(Q_\tau > \alp) \leq P\big[ F_{B(m-A_m^t-A_m^d,\frac{1}{2})}(D_m - A_m^d) \le \gamma  \big].
$$
Recall that $A_m^t$ and $A_m^d$ are fixed.
Hence, by assumption, $X:=D_m - A_m^d$ possesses a binomial $B(m-A_m^t-A_m^d,\frac{1}{2})$ distribution, and $F_X := F_{B(m-A_m^t-A_m^d,\frac{1}{2})}$ is its CDF.
Thus,
$$
P(Q_\tau > \alpha) \leq P(F_X(X)\le \gamma ).
$$
Since, for any random variable $X$, $F_X(X)$ stochastically dominates the uniform (0,1) distribution, it follows that (with $U \sim \text{Unif}(0,1)$),
$$
P(Q_\tau > \alpha) \leq P(F_X(X)\le \gamma ) \leq P(U \le \gamma ) = \gamma .
$$
Hence, even when $j = \infty$, FDP-SD controls the FDP with confidence $1-\gamma$, concluding the proof of Theorem \ref{thm:FDP-SD}.

\end{proof}

\subsection{Simulations of the Spectrum Identification Problem}
\label{supsec:model_simulation}

The spectrum-ID model was presented in Section \ref{sec:background}. Here we used a variant of this model described in \cite{keich:progressive}
where we model the number of candidate peptides a spectrum is compared with, $n$: in practice each spectrum is 
only compared against a subset of peptides in the DB whose mass is within the measurement tolerance of the precursor mass
associated with the spectrum. In this case the $n$ candidate target peptides of a native spectrum are its generating peptide and $n-1$ random peptides,
so $Y_i$ is the best score among $n-1$ such random matches whereas $\til Z_i$ is the best score among $n$ random matches.
It follows that Assumption \ref{fund_assump} is only approximately valid: for native spectra there is a slightly larger chance
a true null will be a decoy win (which creates a slightly conservative --- and hence not overly concerning  --- bias).

We generated simulated instances of the spectrum-ID problem using both calibrated and uncalibrated scores as described next.

\subsubsection{Using Calibrated Scores}

For each of the following nine parameter combinations we generated 40K simulated instances of the spectrum-ID problem
by independently drawing the $X_i$, $Y_i$ and $\til Z_i$ scores for $i=1,\dots,m$, where $m$ is the number of spectra.
We varied $m$ among 500, 2k, and 10k and we varied $\pi_0$, the proportion of foreign spectra, among 0.2, 0.5 and 0.8.
For a native spectrum we drew $X_i$ from a $1-\operatorname{Beta}(a,b)$ distribution (we used $a=0.05$ and $b=10$), and $Y_i$
from a $1-\operatorname{Beta}(1,n-1)$ (we used $n=100$ candidates), whereas for a foreign spectrum we set $X_i=0$.
The $Y_i$ scores for all foreign spectra as well as all the $\til Z_i$ scores were drawn from a  $1-\operatorname{Beta}(1,n)$ distribution
(with the latter ensuring the scores are calibrated).
We then applied TDC, FDP-SD, and FDP-KRB with FDR/FDP thresholds of $\alpha$ = 1\%, 5\%, and 10\%, and confidence levels
$100(1-\gam)$ = 95\% and 99\%.

\subsubsection{Using Uncalibrated Scores}

We generated data with uncalibrated scores as described in \cite{keich:progressive} by associating with each spectrum
a pair of location and scale parameters randomly drawn from a pool of such parameters estimated on a yeast dataset.
We then randomly drew for each spectrum its associated $X_i$, $Y_i$ and $\til Z_i$ scores as in the calibrate case and then we replaced each one
with the corresponding quantile of the Gumbel distribution with the spectrum-specific location and scale parameters.
That is, the inverse of the appropriate Gumbel CDF was applied to each of the three scores. The rest remains the same as in
the calibrated score case.

\subsection{Peptide Detection / Analysis of the ISB18 Dataset}
\label{supsec:isb18_analysis}

We used the same methodology as described in \cite{emery:multipleRECOMB}
for detecting peptides in the ISB18 data set~\cite{klimek:standard}.
Recapped next, this process generated 900 sets of paired target and decoy scores assigned to each peptide in our database.

As in spectrum ID, we first use Tide~\cite{diament:faster} to find for each spectrum its best matching peptide
in the target database as well as in the decoy peptide database.
We then assign to the $i$th target peptide the score, $Z_i$,
which is the maximum of all the PSM scores that were optimally matched to this peptide.
The corresponding decoy score $\til Z_i$ is defined analogously.
We repeat this process using 9 different aliquots, or spectra sets, each paired with 100 randomly shuffled decoys databases
creating 900 sets of paired target and decoy scores to which we applied TDC and FDP-SD
with FDR/FDP thresholds of $\alpha$ = 5\% and confidence a level $100(1-\gam)$ = 95\%.

The ISB18 data set is derived from a series of experiments using an
18-protein standard protein mixture (\url{https://regis-web.systemsbiology.net/PublicDatasets},~\cite{klimek:standard}).
We use 10 runs carried out on an Orbitrap (\texttt{Mix\_7}).

Searches were carried out using the Tide search engine \cite{diament:faster}
as implemented in Crux~\cite{park:rapid}. The peptide database included
fully tryptic peptides, with a static modification for cysteine carbamidomethylation
(\texttt{C+57.0214}) and a variable modification allowing up to six
oxidized methionines (\texttt{6M+15.9949}). Precursor window size
was selected automatically with Param-Medic~\cite{may:param-medic}.
The XCorr score function was employed using
a fragment bin size selected by Param-Medic.

The ISB18 is a fairly unusual dataset in that it was generated using
a controlled experiment, so the peptides that generated the spectra
could have essentially only come from the 18 purified proteins used
in the experiment. We used this dataset to get  feedback on how well our
methods control the FDR/FDP, as explained next.

The spectra set was scanned against a target database that included,
in addition to the 463 peptides of the 18 purified proteins, 29,379
peptides of 1,709 \emph{H.~influenzae} proteins (with ID's beginning
with \texttt{gi|}). The latter foreign peptides were added in order
to help us identify false positives: any foreign peptide reported
is clearly a false discovery. Moreover, because the foreign peptides
represent the overwhelming majority of the peptides in the target
database (a ratio of 63.5 : 1), a native ISB18 peptide reported is
most likely a true discovery (a randomly discovered peptide is much
more likely to belong to the foreign majority). Taken together, this
allows us to gauge the actual FDP for in each reported discovery list.

The 87,549 spectra of the ISB18 dataset were assembled from 10 different
aliquots, so in practice we essentially have 10 independent
replicates of the experiment. However, the last aliquot had only 325
spectra that registered any match against the combined target database,
compared with an average of over 3,800 spectra for the other 9 aliquots,
so we left it out when we independently applied our analysis to each
of the replicates. The spectra set of each of those 9 aliquots was scanned against
the target database paired with each of 100 randomly drawn decoy databases
yielding a total of 900 pairs of target-decoy sets of scores.

\vspace{2em}
\subsection{Supplementary Figures}

\begin{figure}
\centering %
\begin{tabular}{lll}
\hspace{-10ex}
$m=500$  & $m=$2K & $m=$10K \tabularnewline
\hspace{-10ex}
\includegraphics[width=2.5in]{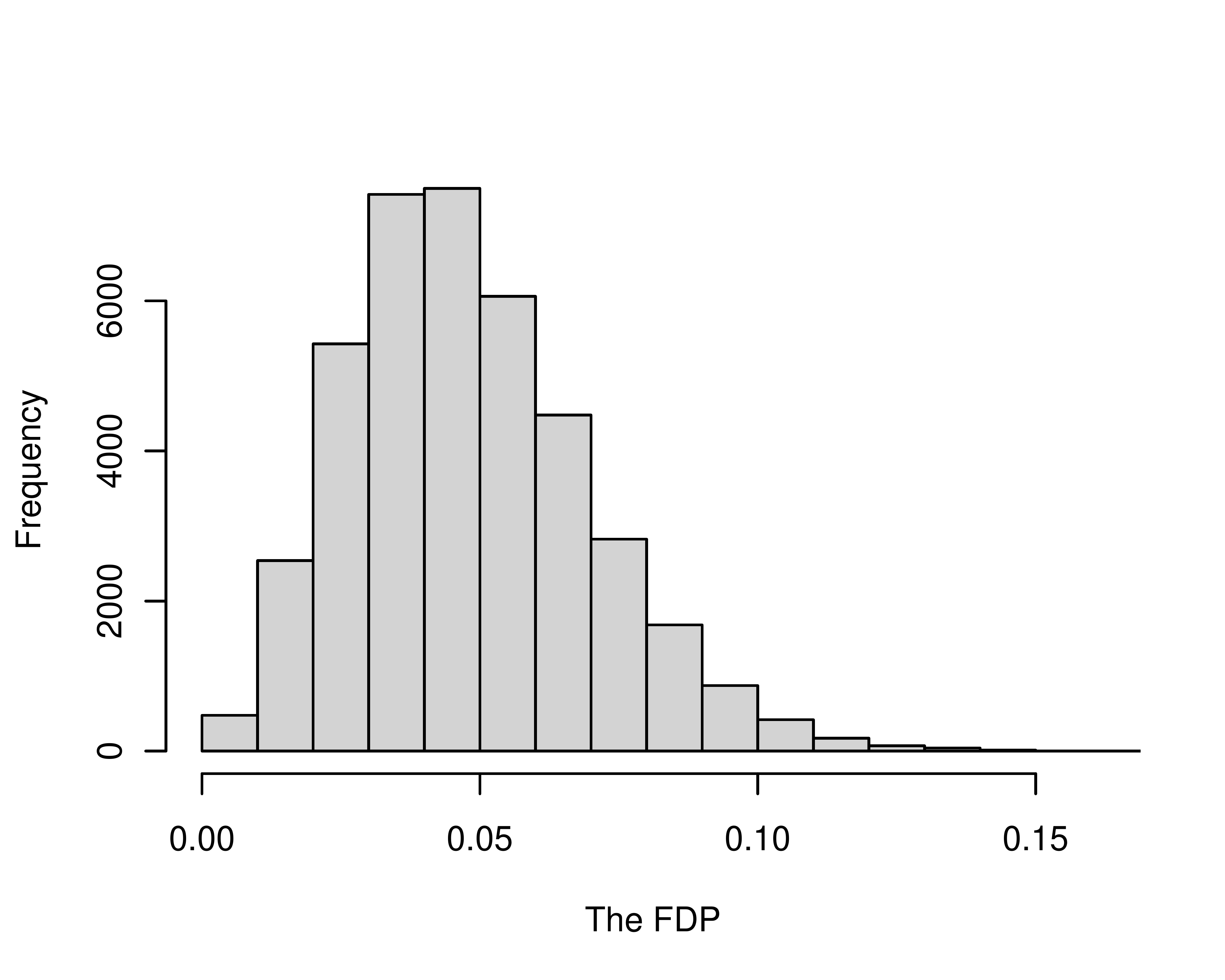}  & \includegraphics[width=2.5in]{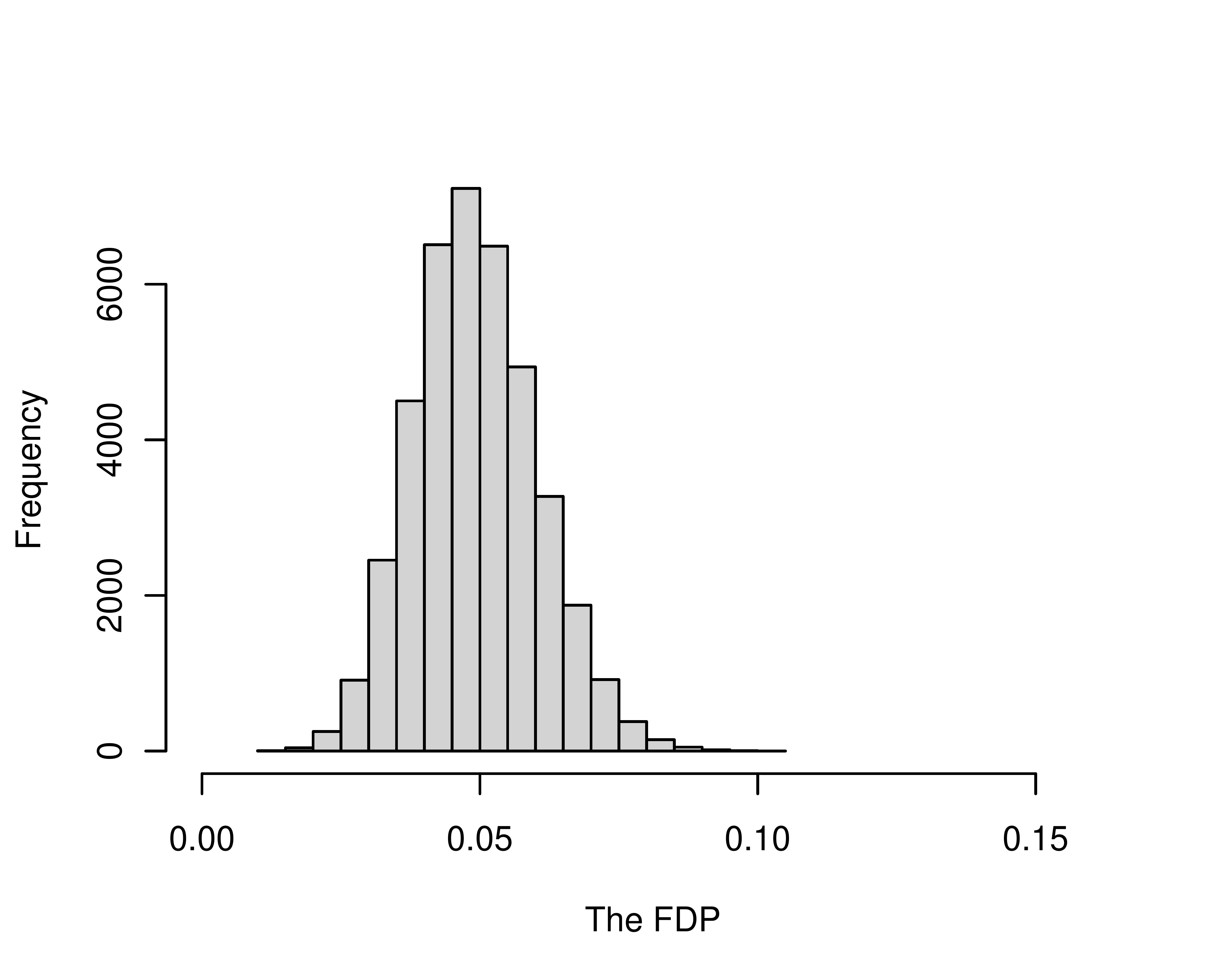}  & \includegraphics[width=2.5in]{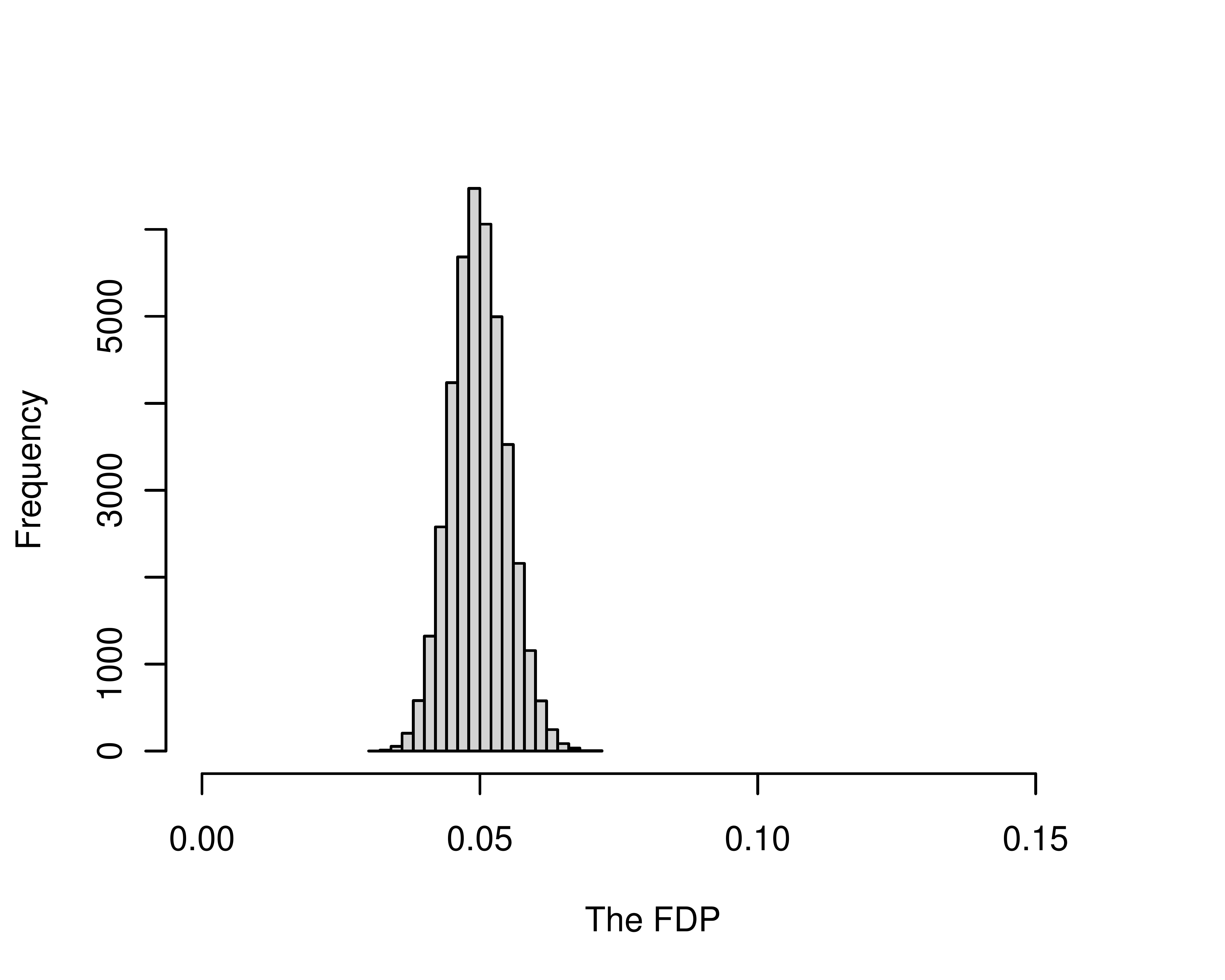} \tabularnewline
\hspace{-10ex}
\includegraphics[width=2.5in]{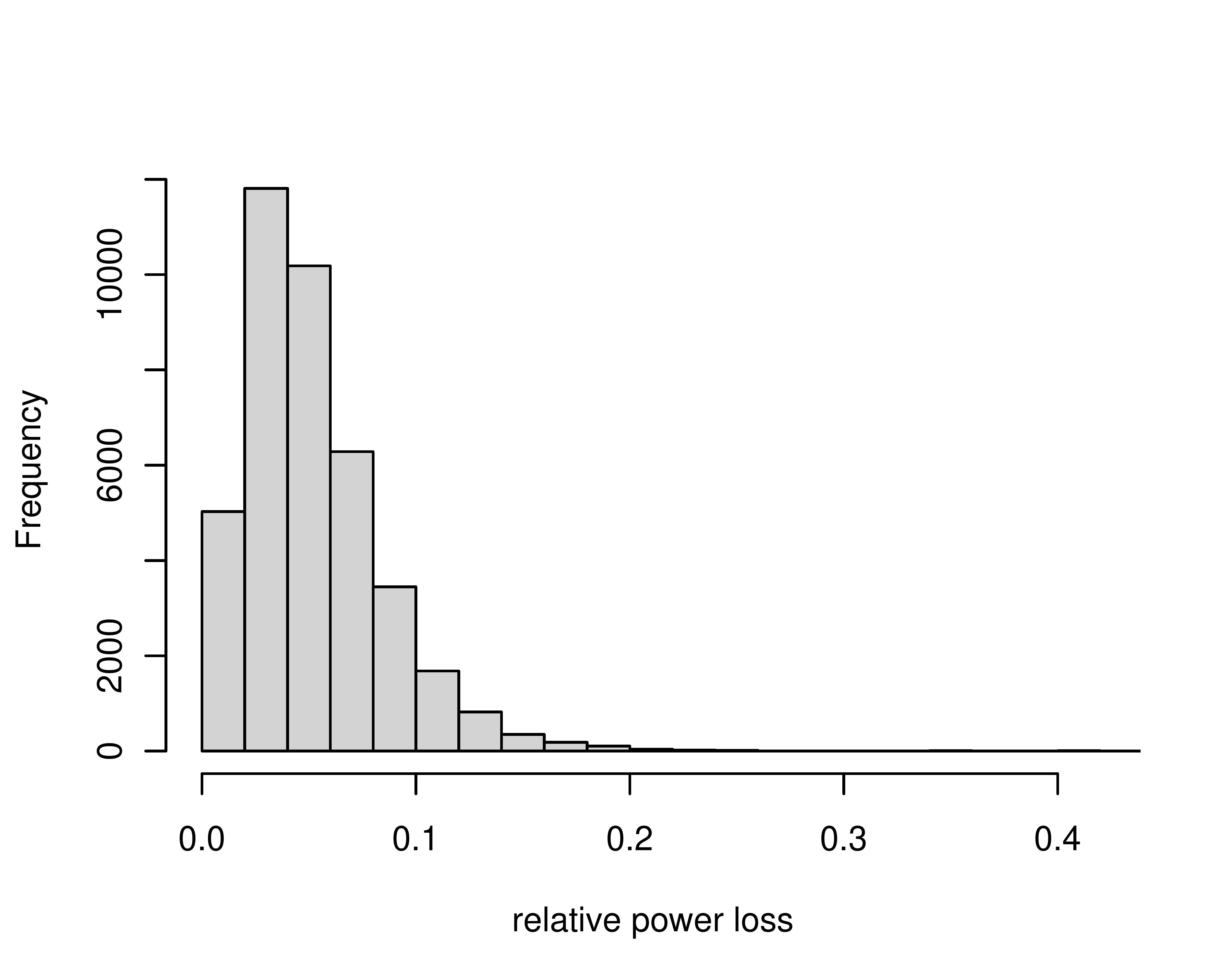}  & \includegraphics[width=2.5in]{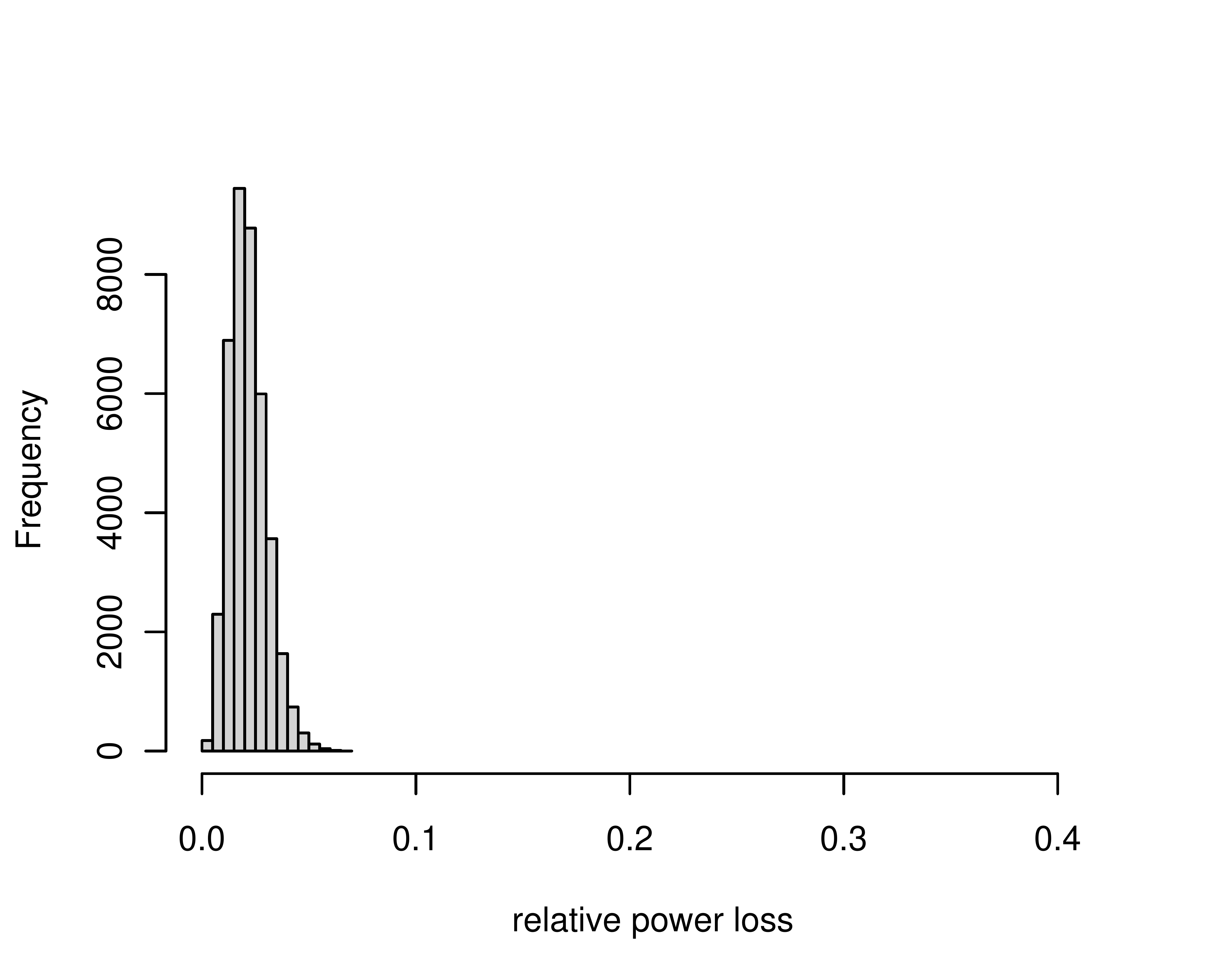}  & \includegraphics[width=2.5in]{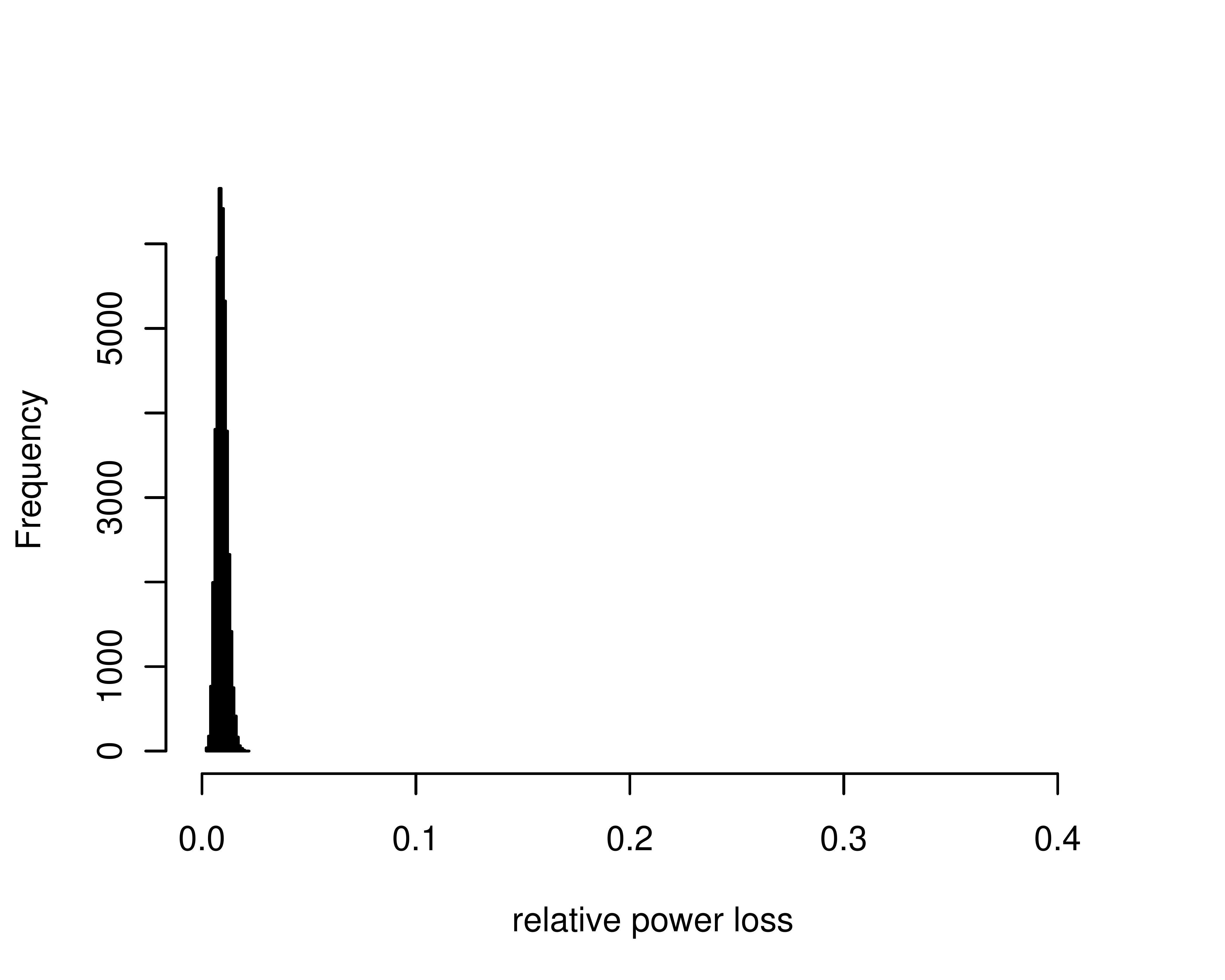} \tabularnewline
\end{tabular}
\caption{\textbf{Varying $m$ in simulated sets of spectrum-ID.} We increase $m$ from 500 (left column) through 2K (middle column)
to 10K (right column) while looking at TDC's FDP (top row), the relative loss of power (in terms of the number of correct discoveries) when using FDP-SD compared with TDC (middle row).
The other parameters were kept constant: $\alp=0.05,\gam=0.05,\pi_0=0.5$.
\label{supfig:vary_m}}
\end{figure}

\begin{figure}
\centering %
\begin{tabular}{lll}
\hspace{-10ex}
$\pi_0=0.8$  & $\pi_0=0.5$ & $\pi_0=0.2$ \tabularnewline
\hspace{-10ex}
\includegraphics[width=2.5in]{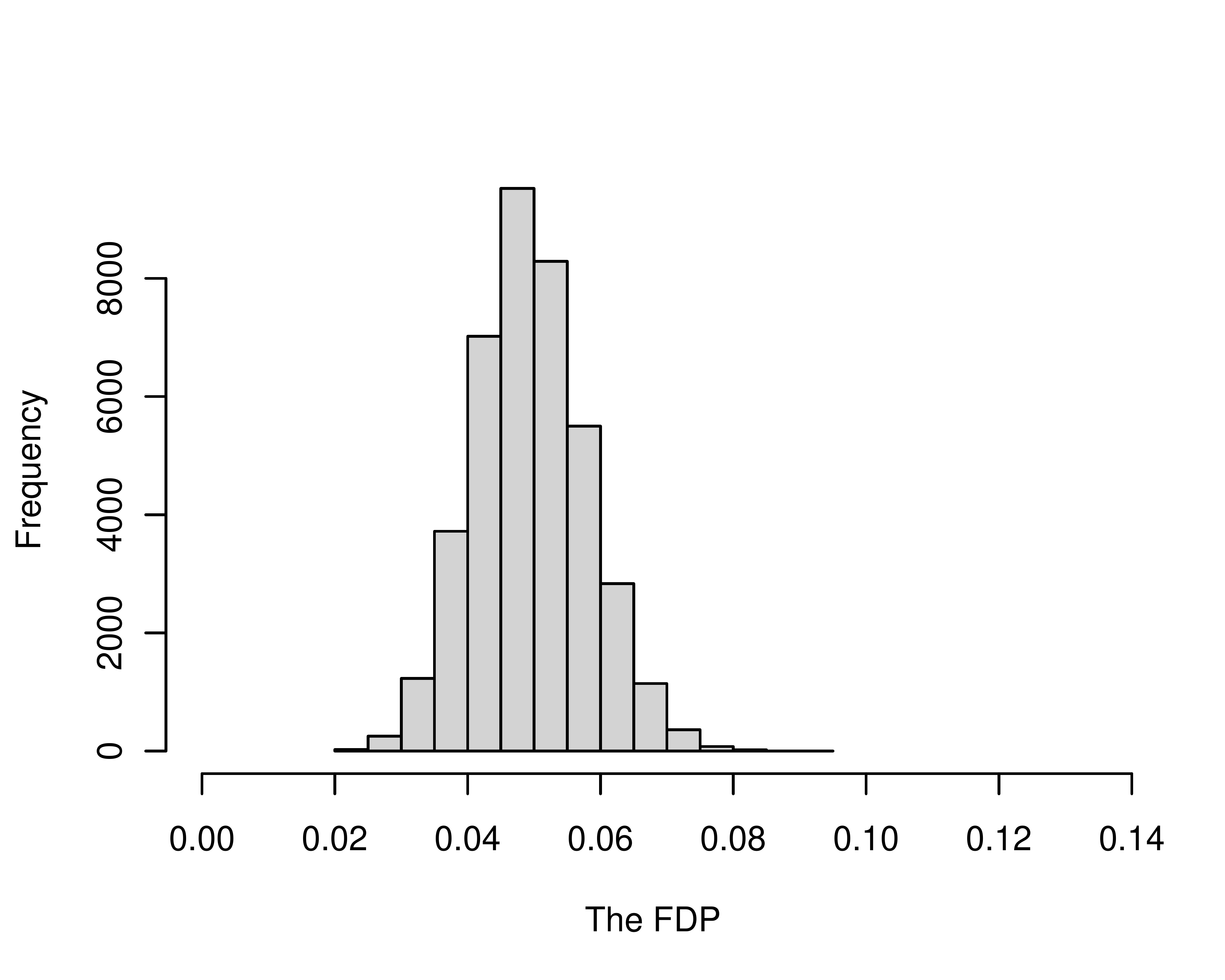}  & \includegraphics[width=2.5in]{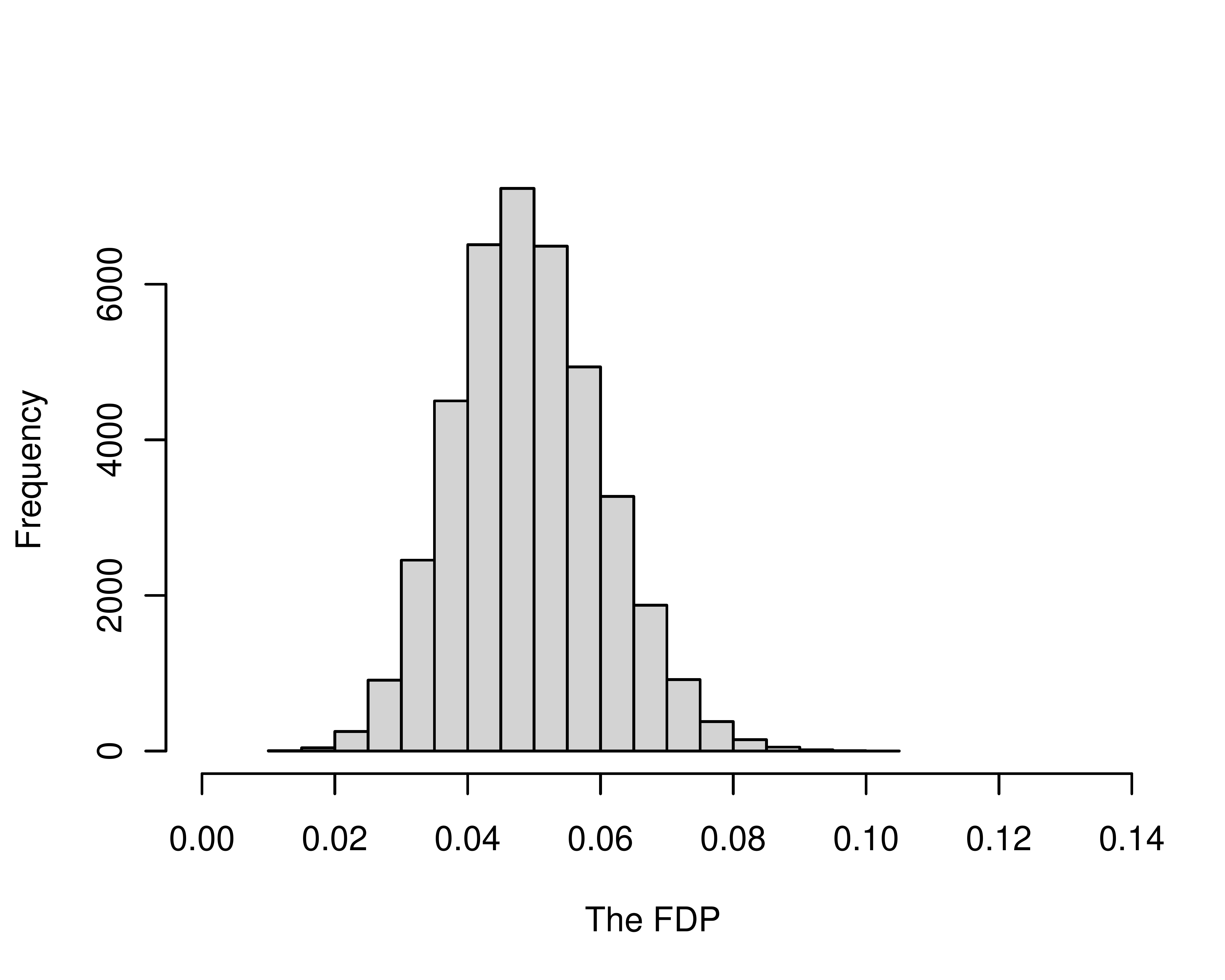}  & \includegraphics[width=2.5in]{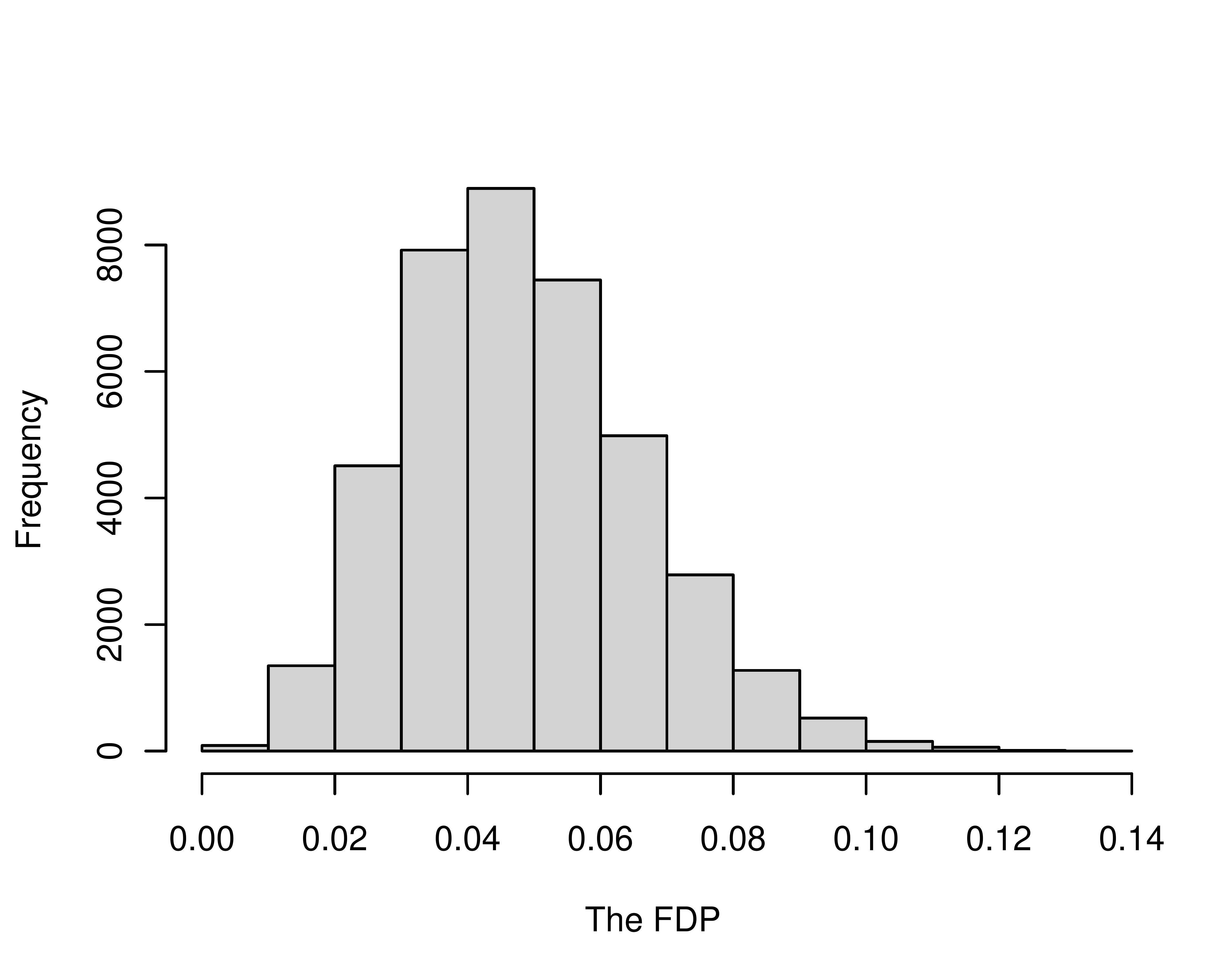} \tabularnewline
\hspace{-10ex}
\includegraphics[width=2.5in]{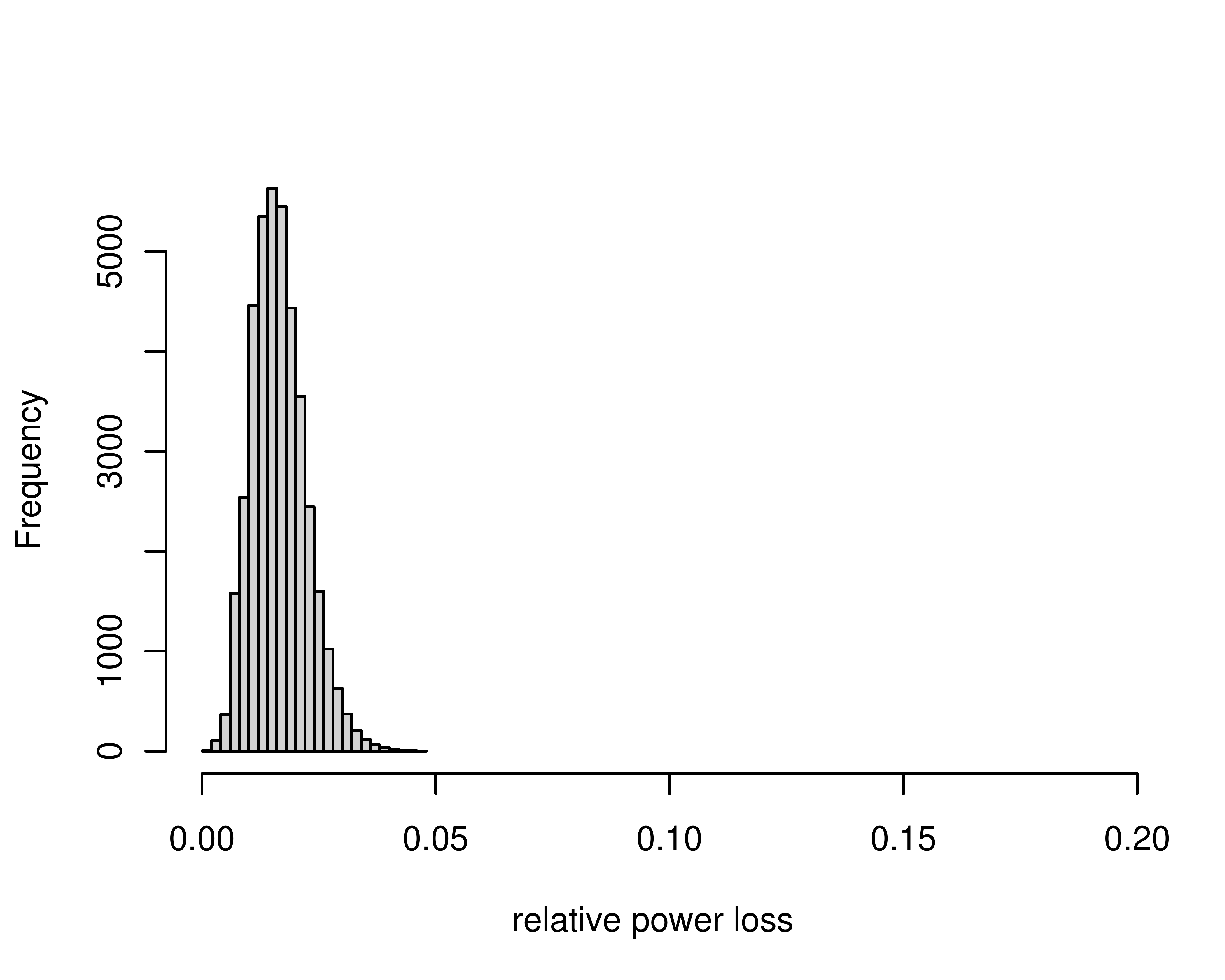}  & \includegraphics[width=2.5in]{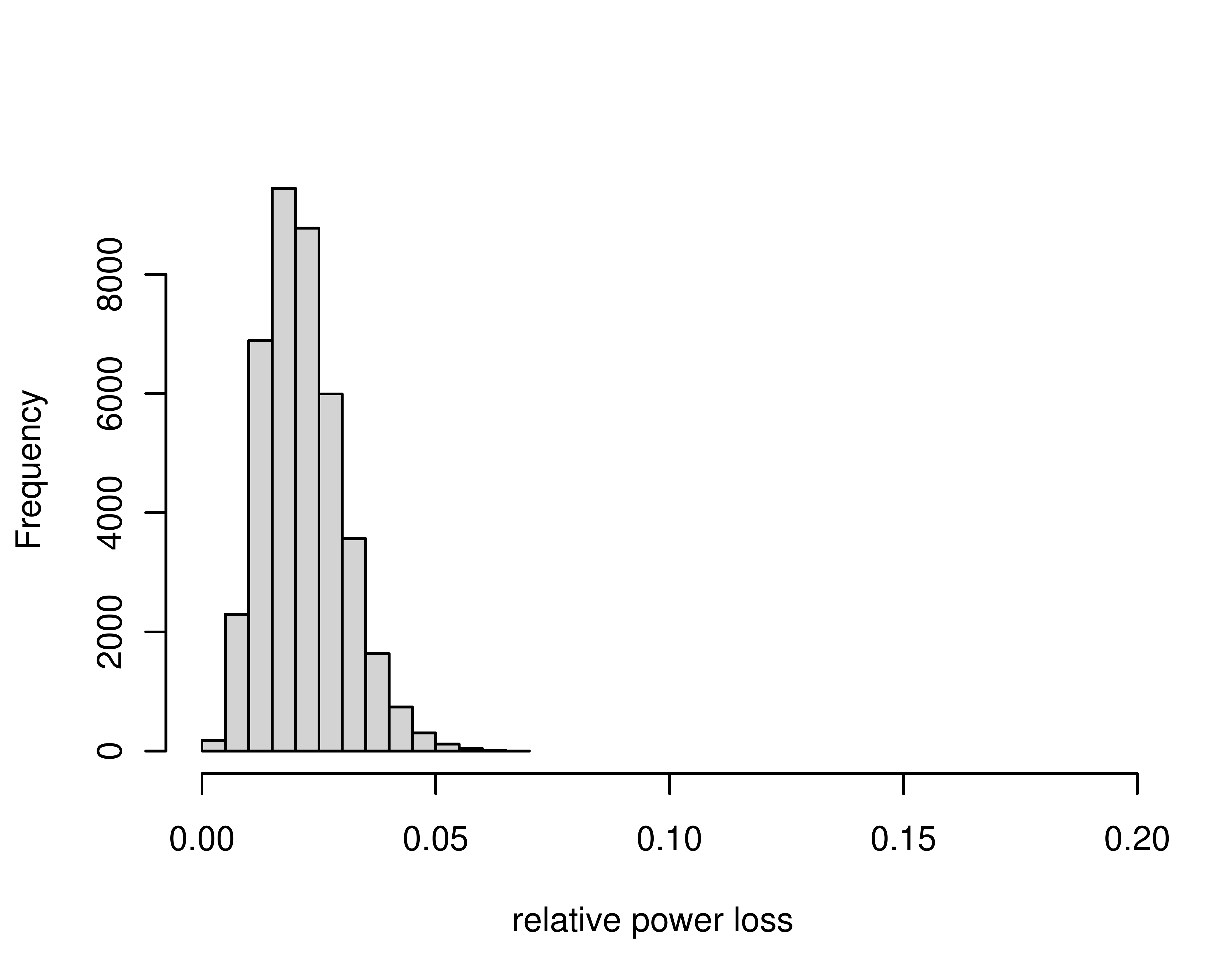}  & \includegraphics[width=2.5in]{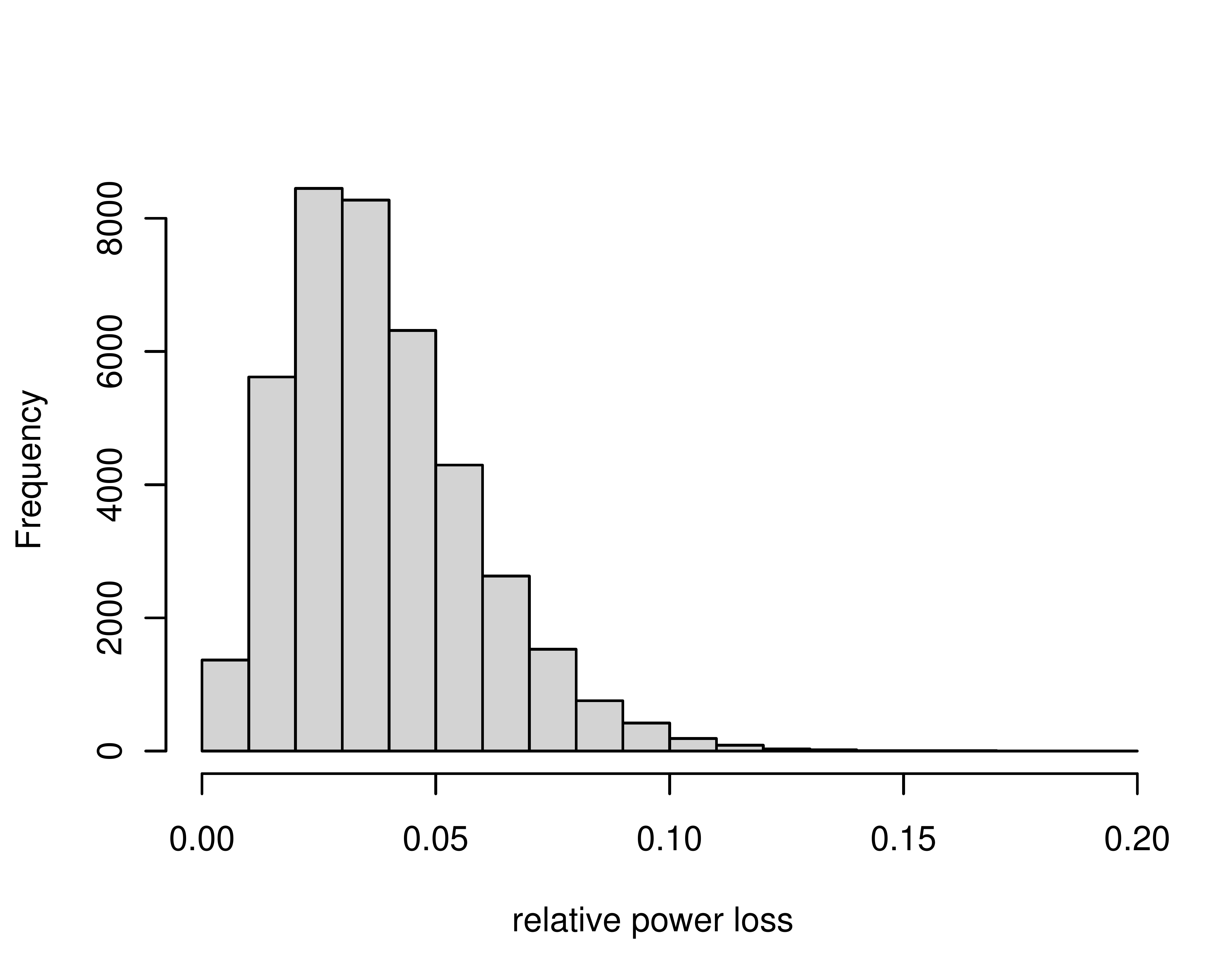} \tabularnewline
\end{tabular}
\caption{\textbf{Varying $\pi_0$ in simulated sets of spectrum-ID.} Similar to \supfig~\ref{supfig:vary_m} only here we decrease $\pi_0$
keeping $\alp=0.05$, $\gam=0.05$, and $m=$2K.  TDC's FDP (top row), the relative loss of power when using FDP-SD compared with TDC (middle row).
\label{supfig:vary_pi0}}
\end{figure}

\clearpage

\bibliography{refs}
\bibliographystyle{plain}

\end{document}